\documentclass[letterpaper, conference]{ieeetran}
\usepackage[OT1]{fontenc}

\IEEEoverridecommandlockouts  
\usepackage{settings/preamble}

\renewcommand{\P}{\mathbb{P}}
\newcommand{\cP}{\mathcal{P}}

\newcommand{\matO}{\mathbf{0}}
\renewcommand{\S}{\mathcal{S}}
\newcommand{\K}{\mathcal{K}}
\newcommand{\bcvar}{\textrm{-} \mathrm{CVaR}}

\newcommand{\sol}{\textrm{sol}}
\newcommand{\lqr}{\textrm{lqr}}
\newcommand{\wcvar}{\textrm{W-CVaR}}

\newcommand{\bx}{\bar{x}}
\newcommand{\bu}{\bar{u}}

\newcommand{\bmu}{\bar{\mu}}
\newcommand{\bsigma}{\bar{\sigma}}
\newcommand{\bSigma}{\bar{\varSigma}}
\newcommand{\bOmega}{\bar{\varOmega}}
\newcommand{\bGamma}{\bar{\varGamma}}
\newcommand{\bS}{\bar{S}}

\newcommand{\wmu}{\mu_w}
\newcommand{\wSigma}{\varSigma_w}
\newcommand{\wsigma}{\sigma_w}

\newcommand{\tM}{\tilde{M}}
\newcommand{\tm}{\tilde{m}}
\newcommand{\tp}{\tilde{p}}
\newcommand{\tOmega}{\tilde{\varOmega}}
\newcommand{\tomega}{\tilde{\omega}}
\newcommand{\tq}{\tilde{q}}
\newcommand{\tmu}{\tilde{\mu}}

\newcommand{\bxi}{\bm{\xi}}
\newcommand{\btheta}{\bm{\theta}}

\newcommand{\JI}{J^{\mathrm{I}}}
\newcommand{\JII}{J^{\mathrm{II}}}

\newcommand{\boldtitle}[1]{\medskip\noindent\textbf{#1}}

\renewcommand{\succeq}{\succcurlyeq}

\newcommand{\svdots}{\raisebox{0pt}{\scalebox{.58}{$\vdots$}}}
\newcommand{\sddots}{\raisebox{0pt}{\scalebox{.58}{$\ddots$}}}

\usepackage{environ}

\NewEnviron{equationfit}[1][0.9\linewidth]{%
\begin{equation*}
  \resizebox{#1}{!}{\(
    \BODY
  \)}
\end{equation*}
}
\NewEnviron{alignfit}[1][0.9\linewidth]{%
\begin{equation*}
  \resizebox{#1}{!}{\(
  \begin{aligned}
    \BODY
  \end{aligned}
  \)}
\end{equation*}
}

\newcommand{\shahriar}[1]{\textcolor{violet}{Shahriar:#1}}

\begin{document}

  \allowdisplaybreaks[4]
  \setlength{\abovedisplayskip}{1pt}
  \setlength{\abovedisplayshortskip}{1pt}
  \setlength{\belowdisplayskip}{1pt}
  \setlength{\belowdisplayshortskip}{1pt}
  \setlength{\jot}{1pt}  
  \setlength{\floatsep}{1ex}
  \setlength{\textfloatsep}{1ex}



  \title{\vspace{0.5cm}Risk-sensitive Affine Control Synthesis for\\Stationary LTI Systems
}

\author{Yang Hu, Shahriar Talebi, Na Li
\thanks{The authors are with the School of Engineering and Applied Sciences, Harvard University, USA. This work is supported by NSF AI Institute 2112085. Emails: yanghu@g.harvard.edu, talebi@g.harvard.edu, nali@seas.harvard.edu.}
}

\maketitle

\begin{abstract}
  To address deviations from expected performance in stochastic systems, we propose a risk-sensitive control synthesis method to minimize certain risk measures over the limiting stationary distribution. Specifically, we extend Worst-case Conditional Value-at-Risk (W-CVaR) optimization for Linear Time-invariant (LTI) systems to handle nonzero-mean noise and affine controllers, using only the first and second moments of noise, which enhances robustness against model uncertainty. Highlighting the strong coupling between the linear and bias terms of the controller, we reformulate the synthesis problem as a Bilinear Matrix Inequality (BMI), and propose an alternating optimization algorithm with guaranteed convergence. Finally, we demonstrate the numerical performance of our approach in two representative settings, which shows that the proposed algorithm successfully synthesizes risk-sensitive controllers that outperform the naive LQR baseline.
\end{abstract}

  \section{Introduction}\label{sec:1-introduction}

In stochastic systems, optimizing solely for expected performance metrics can yield unsatisfactory outcomes. Instead, it is often necessary to consider higher-order characteristics of the distribution to address deviations from average performance that arises from the inherent stochasticity. Risk-sensitive modeling, originally developed in finance, offers systematic techniques to account for these deviations under uncertainty, with broad applicability across a wide range of engineering fields, including risk-sensitive control \cite{duncan2013linear, moon2020generalized}, risk-constrained control \cite{tsiamis2020risk, zhao2021infinite}, risk-sensitive model predictive control \cite{singh2018framework, sopasakis2019risk}, and risk-sensitive reinforcement learning \cite{wang2023near, zhang2023regularized}.

To quantify the risks observed in stochastic systems, various measures are introduced to capture distributional properties of random variables \cite{delbaen2002coherent, follmer2002convex}, among which the most widely-used are value-at-risk (VaR) \cite{duffie1997overview}, conditional value-at-risk (CVaR) \cite{rockafellar2000optimization} and entropic risk measure \cite{follmer2011entropic}. However, computing these risk measures often requires full knowledge of the governing distributions, which is not necessarily available in practice due to model uncertainty and limited data samples. \textit{Worst-Case CVaR (W-CVaR)} \cite{zymler2013distributionally}), as its name suggests, measures the CVaR risk in the worst case within the family of distributions with the same prescribed first and second moments. As an upper bound of the actual risk measured in CVaR, it is not only capable of inducing risk-sensitive performance in face of limited information, but is also robust against model uncertainty by only requiring the first two moments of the noise---which is standard for stochastic systems \cite{van2015distributionally, kishida2023risk}.

In stochastic Linear Time-Invariant (LTI) systems, we often optimize cumulative step-wise performance measures which, if designed carefully, also implies favorable properties regarding long-term closed-loop behavior (e.g., asymptotic stability, bounded second moments, etc.). The (stationary) limiting distribution of the process, however, provides a more comprehensive view by capturing this long-term behavior under sustained stochastic disturbances. This viewpoint is not only of theoretical interest \cite{collins1987PhDthesis, chen2015optimal, liu2022optimal}, but also crucial in scenarios like molecular engineering \cite{toyabe2010nonequilibrium, braiman2003control} and spacecraft control \cite{zhang2023stochastic, benedikter2022covariance}, where short-term fluctuations are often manageable separately, but the accumulated long-term imbalances will eventually impair the output quality. 

Most recently, the worst-case CVaR measure on the stationary distribution of LTI systems (so-called ``stationary LTI systems'') was considered in \cite{kishida2023risk}, but only under the assumptions of \textit{zero-mean} noise and \textit{linear} stationary controllers. In this special case, the W-CVaR-optimal controller happens to coincide with the LQR optimal controller. However, in practice, noise may have a nonzero-mean representing a constant drift (see, for example, the HVAC dynamics in \Cref{sec:apdx-simulation_setting}), which necessitates the extension of the analysis to account for \textit{nonzero-mean} noise and \textit{affine} controllers. Then naturally, one might wonder whether, similar to the LQR controller synthesis, the optimization over the linear term can be decoupled from the bias term in the risk-sensitive affine control synthesis problem. Unfortunately, this is not the case---there exists a strong coupling between the linear and bias terms of the controller, due to the interaction between the mean and covariance of the stationary distribution in the risk-sensitive synthesis. As shown in the example provided in \Cref{sec:example_1d}, the W-CVaR-optimal (affine) controller becomes increasingly advantageous over the so-called \textit{naive affine LQR} controller that, by ignoring the coupling, assigns the bias and linear term separately. This phenomenon also illustrates the increasing coupling between the linear and bias terms as the mean and covariance of the noise grow larger.


\boldtitle{Contributions.} In this paper, we consider stationary LTI systems with \textit{nonzero-mean} noise, and formulate a risk-sensitive \textit{affine} control synthesis problem in \Cref{sec:2-settings} that optimizes a W-CVaR measure on the limiting stationary distribution. We also provide an explanatory example to illustrate the necessity of jointly optimizing the linear and bias terms. To find the risk-sensitive affine controller, we reformulate the problem as a Bilinear Matrix Inequality (BMI) in \Cref{sec:3-algorithm}, which not only captures the strong coupling between the linear and bias terms, but also allows for the design of an alternating optimization algorithm to solve it. The algorithm alternates between two convex sub-problems that are easier to solve numerically, and is guaranteed to converge to a suboptimal solution that can be shown by leveraging the monotonicity and boundedness of the generated sequence (see \Cref{sec:convergence} for details). Finally, we demonstrate the effectiveness of our algorithm in \Cref{sec:4-simulations} by numerically synthesizing risk-sensitive controllers in two representative settings, Inverted Pendulum \cite{hespanha2018linear} and Multi-zone Heating, Ventilation, and Air Conditioning (HVAC) Thermal Control \cite{li2021distributed}. For better presentation, some technical proofs are deferred to the appendix.

\boldtitle{Notations.} Let $A \succ B$ ($A \succeq B$) denote the fact that $A-B$ is positive (semi)definite. Let $\S^n$ denote the set of symmetric matrices, where subscripts may be used to indicate positive-(semi-)definiteness. Any symmetric matrix can be represented by its upper triangular part only, with the lower triangular part filled by $*$. Let $\rho(A)$ denote the spectral radius of a matrix $A$. Let $\norm{\cdot}$ denote the Euclidean 2-norm of vectors and matrices. Let $\angl{A, B} := \tr(A^{\top} B)$ denote the standard Euclidean inner product of matrices. Let $A^{\dagger}$ denote the Moore-Penrose pseudoinverse of a matrix $A$. Let $(x)^+ := \max\brac{x, 0}$. Let $\Delta(S)$ denote the set of all probability measures on a set $S$. Let $\texttt{dlqr}(A,B,Q,R) := -(R + B^{\top} P B)^{-1} B^{\top} P A$ denote the optimal controller for the LQR instance $(A,B,Q,R)$, where $P$ is the solution to the corresponding discrete-time algebraic Riccati equation (DARE).

  \section{Preliminaries and Problem Formulation}\label{sec:2-settings}

We consider an infinite-horizon discrete-time LTI system
\begin{equation*}
  x_{t+1} = A x_t + B u_t + w_t,~ \forall t \in \N,
\end{equation*}
where $x_t \in \R^n$ is the state vector, and $u_t \in \R^m$ is the input vector; $w_t \in \R^d$ is an i.i.d. noise independent of $x_t$, $u_t$ and time step $t$, with mean $\wmu \in \R^n$ and covariance matrix $\wSigma \in \S^{n}_{\succeq 0}$; Matrices $A \in \R^{n \times n}$ and $B \in \R^{n \times m}$ denote system parameters. We point out that the noise covariance $\wSigma$ is not necessarily required to be positive definite, and make the following minimal assumption regarding the dynamics.

\begin{assumption}
  $(A,B)$ is stabilizable.
\end{assumption}

To properly handle the nonzero-mean setting, instead of the standard choice of linear feedback controllers, we extend the scope of candidate controllers to stabilizing \textit{affine} feedback controllers $u = Kx + \ell$, where the \textit{linear} part $K \in \K : =\set{K \in \R^{m \times n} \mid \rho(A + BK) < 1}$ is assumed to stabilize the closed-loop system, and $\ell \in \R^m$ is called the \textit{bias} term. We further discuss the reason for considering the extended class of affine controllers in \Cref{sec:3-1-general_mean_challenge}.

As our choice of risk measure is related to the class of distributions with prescribed mean and covariance, we characterize the closed-loop limiting first and second moments under affine feedback controllers in the following lemma.

\begin{lemma}\label{thm:limiting_distribution}
  With a stabilizing affine feedback controller $u = Kx + \ell$ (i.e., $K \in \K$), the closed-loop system $x_{t+1} = (A+BK) x_t + B \ell + w_t$
  admits a unique set of first two limiting moments $(\bmu, \bSigma)$, which are specified by
  \begin{align*}
    \bmu &= \sum_{\tau = 0}^{\infty} (A+BK)^{\tau} (B\ell + \wmu), \\
    \bSigma &= \sum_{\tau=0}^{\infty} (A+BK)^{\tau} \wSigma ((A+BK)^{\tau})^{\top}.
  \end{align*}
  Further, these limiting moments are invariant under the closed-loop dynamics, i.e.,
  \begin{subequations}\label{eq:stationary_distribution}
  \begin{align}
    (A+BK) \bmu + B\ell + \wmu &= \bmu,\\
    (A+BK) \bSigma (A+BK)^{\top} + \wSigma &= \bSigma. \label{eq:stationary_distribution:Sigma}
  \end{align} 
  \end{subequations}
\end{lemma}

\subsection{Risk Measure: Worst-case CVaR}\label{sec:settings-WCVaR}
Next, we define worst-case CVaR (W-CVaR) as the risk measure that is quantified by the largest CVaR value among a family of distributions.

\begin{definition}[CVaR and W-CVaR \cite{rockafellar2000optimization,zymler2013distributionally}]
  Given a risk level $\beta \in (0,1)$, for any distribution $\P \in \Delta(\varXi)$ and random variable $L: \varXi \to \R$ defined on a sample space $\varXi$, define the CVaR of $L(\xi)$ with respect to $\xi \sim \P$ at risk level $\beta$ as%
  \footnote{Here we follow the definition in \cite{rockafellar2000optimization}. Some literature may refer to $\varepsilon := 1-\beta$ as the risk level.}%
  \begin{equation*}
    \P\bcvar_{\beta}[L(\xi)] := \min_{\alpha \in \R} \brac*{\alpha + \tfrac{1}{1-\beta} \E[\xi \sim \P]{(L(\xi)-\alpha)^+}}.
  \end{equation*}
  Further, for a family of distributions $\cP \subseteq \Delta(\varXi)$ over $\varXi$, define the W-CVaR of $L(\xi)$ with respect to $\cP$ at risk level $\beta$ as
  \begin{equation*}
    \cP\bcvar_{\beta}[L(\xi)] := \sup_{\P \in \cP} \P\bcvar_{\beta}[L(\xi)].
  \end{equation*}
\end{definition}

Since the stochastic noise is only modelled up to its first and second moments, the state distributions are also only specifiable to the first two moments. Therefore, in this paper, we focus on the family $\cP_{\mu,\varSigma}$ of distributions with prescribed mean $\mu$ and variance $\varSigma$. It is known that the W-CVaR with respect to $\cP_{\mu,\varSigma}$ can be equivalently characterized by a constrained convex optimization problem (see \Cref{thm:quadratic_worst_CVaR_characterization} in  \Cref{sec:apdx-WCVaR_lemma}).

\subsection{The Risk-sensitive Affine Control Synthesis Problem}

Let $\bx$ obey the limiting state distribution, and $\bu = K \bx + \ell$ be the corresponding control input. Consider the quadratic cost on the limiting state $\bx$ defined by
\begin{align*}
  c_{K,\ell}(\bx) &= \bx^{\top} Q \bx + \bu^{\top} R \bu \\
  &= \bx^{\top} (Q + K^{\top} R K) \bx + 2 (K^{\top} R\ell)^{\top} \bx + \ell^{\top} R \ell,
\end{align*}
where $Q, R \in \S^{n}_{\succ 0}$. Our objective is to synthesize a stabilizing affine controller $u = Kx + \ell$ that minimizes the $\cP_{\bmu, \bSigma}$-W-CVaR measure of the random variable $c(\bx)$, i.e.,
\begin{subequations}\label{eq:problem-original}
\begin{align}
  \min_{\bmu, \bSigma, K, \ell}\quad& \cP_{\bmu, \bSigma}\bcvar[c_{K,\ell}(\bx)] \label{eq:problem-original:1}\\
  \mathrm{s.t.}\quad
    & K \in \K,~ \ell \in \R^m;~ \bmu \in \R^n,~ \bSigma \in \mathcal{S}^{n}_{\succeq 0}, \\
    & (A+BK) \bSigma (A+BK)^{\top} + \wSigma = \bSigma, \\
    & (A + BK) \bmu + B\ell + \wmu = \bmu .
\end{align}
\end{subequations}
By \Cref{thm:quadratic_worst_CVaR_characterization}, we can equivalently rewrite the objective function as a constrained optimization:
\begin{align} \label{eq:WCVaR_equivalence}
  &\cP_{\bmu,\bSigma}\bcvar_{\beta}[c_{K,\ell}(\bx)] ={} \\
  &\min_{ \alpha \in \R, M \in \mathcal{S}^{n+1}_{\succeq 0} } \brac*{\alpha + \tfrac{1}{1-\beta} \angl{\bOmega, M} \;\middle|\; M \succeq \begin{bsmallmatrix}
      Q + K^{\top} R K & K^{\top} R\ell \\
      * & \ell^{\top} R \ell -\alpha
    \end{bsmallmatrix}}, \nonumber
\end{align}
where
\(
  \bOmega := \begin{bsmallmatrix}
    \bSigma + \bmu \bmu^{\top} & \bmu \\
    * & 1
  \end{bsmallmatrix}
\)
denotes the second-order moment matrix of the limiting distribution. Therefore, \eqref{eq:problem-original} is equivalent to the following optimization problem:
\begin{subequations}\label{eq:problem-main}
\begin{align}
  \min_{\alpha, M, \bmu, \bSigma, K, \ell}\quad& \alpha + \tfrac{1}{1-\beta} \angl{\bOmega, M} \label{eq:problem-main:1}\\
  \mathrm{s.t.}\quad
    & M \in \mathcal{S}^{n+1}_{\succeq 0},~ K \in \K,~ \bSigma \in \mathcal{S}^{n}_{\succeq 0}, \label{eq:problem-main:2}\\
    & (A+BK) \bSigma (A+BK)^{\top} + \wSigma = \bSigma, \label{eq:problem-main:3}\\
    & \bmu = (A + BK) \bmu + B\ell + \wmu, \label{eq:problem-main:4}\\
    & M \succeq \begin{bmatrix}
      Q + K^{\top} R K & K^{\top} R\ell \\
      * & \ell^{\top} R \ell -\alpha
    \end{bmatrix}, \label{eq:problem-main:5}\\
    & \bOmega = \begin{bmatrix}
      \bSigma + \bmu \bmu^{\top} & \bmu \\
      * & 1
    \end{bmatrix}. \label{eq:problem-main:6}
\end{align}
\end{subequations}

\subsection{An Example: Naive Affine LQR Controller Is Inferior} \label{sec:example_1d}
Before we propose our synthesis method, one may naturally construct a naive controller as an approximate solution to \eqref{eq:problem-main} based on the following idea: decouple the bias $\ell$ from $K$ by first canceling out $\wmu$ with $B \ell$, i.e. take $\ell_{\lqr} = -B^\dagger \wmu$, to reduce to the zero-mean setting, and then apply the known result that the optimal controller in this case coincides with the LQR controller $K_{\lqr} = \texttt{dlqr}(A,B,Q,R)$ \cite{kishida2023risk}. We refer to $u = K_{\lqr} x + \ell_{\lqr}$ as the \textit{naive affine LQR controller}.

Unfortunately, we can provide the following simple example to illustrate why the performance of the naive affine LQR controller is inferior; in other words, the coupling between the bias and linear terms of the controller are too strong to be ignored---as long as $\wmu \neq \matO$.

Consider a 1-dimensional system with dynamics $x_{t+1} = x_t + u_t + w_t$ ($x_t,u_t,w_t \in \R$).%
\footnote{Since everything is a scalar in this case, we will use the corresponding lower-case letters $k$, $\bsigma$ and $m_1$ to denote $K$, $\bSigma$ and $M_1$ in this example.} %
Set $\beta = \frac{1}{2}$ and $Q = R = 1$. In this case, the limiting moments are simply $\bmu = \frac{\ell + \mu_w}{-k}$, $\bsigma = \frac{\sigma_w}{-k(k+2)}$, so \eqref{eq:problem-main} can be simplified to:
\begin{subequations}\label{eq:problem-example_1d}
\begin{align}
  \min_{\alpha, m_1, p, m_0, k, \ell}\quad& \alpha + 2 \prn[\big]{ m_1 (\bsigma + \bmu^2) + 2p\bmu + m_0 } \label{eq:problem-example_1d:1}\\
  \mathrm{s.t.}\quad
    & -2 < k < 0, \label{eq:problem-example_1d:2}\\
    & \begin{bmatrix}
      m_1 & p \\
      p & m_0
    \end{bmatrix} \succeq \begin{bmatrix}
      k^2 + 1 & k\ell \\
      k\ell & \ell^2 -\alpha
    \end{bmatrix}. \label{eq:problem-example_1d:3}
\end{align}
\end{subequations}
We point out that, despite its low dimensionality and seemingly simple form, \eqref{eq:problem-example_1d} is actually hard to solve using a numerical solver---the product term $2p\bmu$ in the objective \eqref{eq:problem-example_1d:1}, along with the off-diagonal terms $p$ and $k\ell$ in \eqref{eq:problem-example_1d:3}, intensify the coupling in the problem, making it non-convex and thus computationally hard. For the purpose of exposition, we use grid search over $(k, \ell)$ to approximate the optimal solution $(k^{\star}, \ell^{\star})$ of \eqref{eq:problem-example_1d} under different noise settings $(\wmu, \wsigma) = (\mu, \mu)$. The results can be found in \Cref{fig:example_1d_controller} below, where we compare the W-CVaR-optimal affine controller $u = k^{\star} x + \ell^{\star}$ against the naive affine LQR controller $u = k_{\lqr} x + \ell_{\lqr}$ specified above.

\begin{figure}[t]
  \centering
  \begin{subfigure}[b]{0.49\linewidth}
    \centering
    \includegraphics[width=0.99\linewidth]{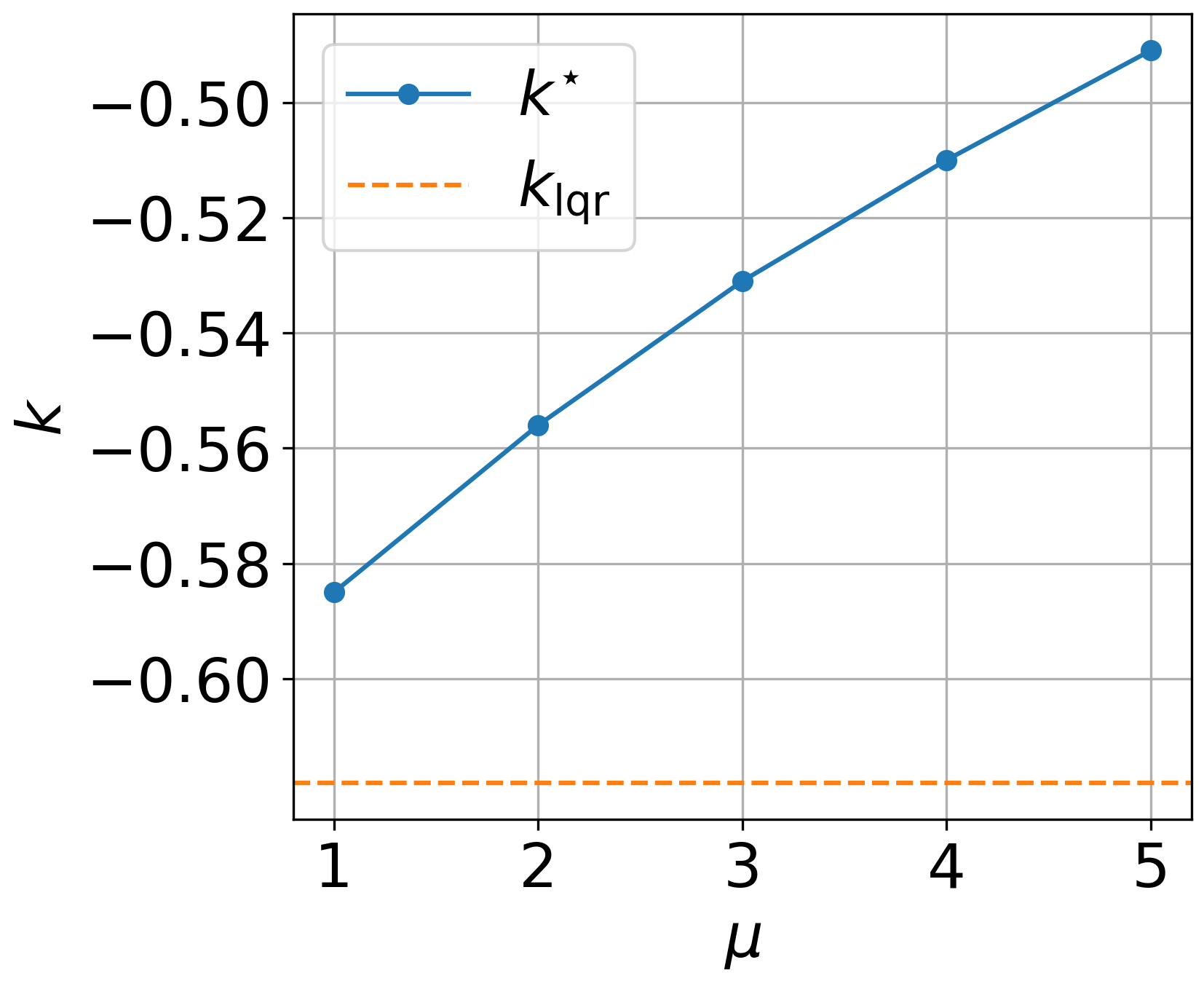}
    \caption{control gain $k$} \label{fig:example_1d_controller:k}
  \end{subfigure}
  \begin{subfigure}[b]{0.49\linewidth}
    \centering
    \includegraphics[width=0.99\linewidth]{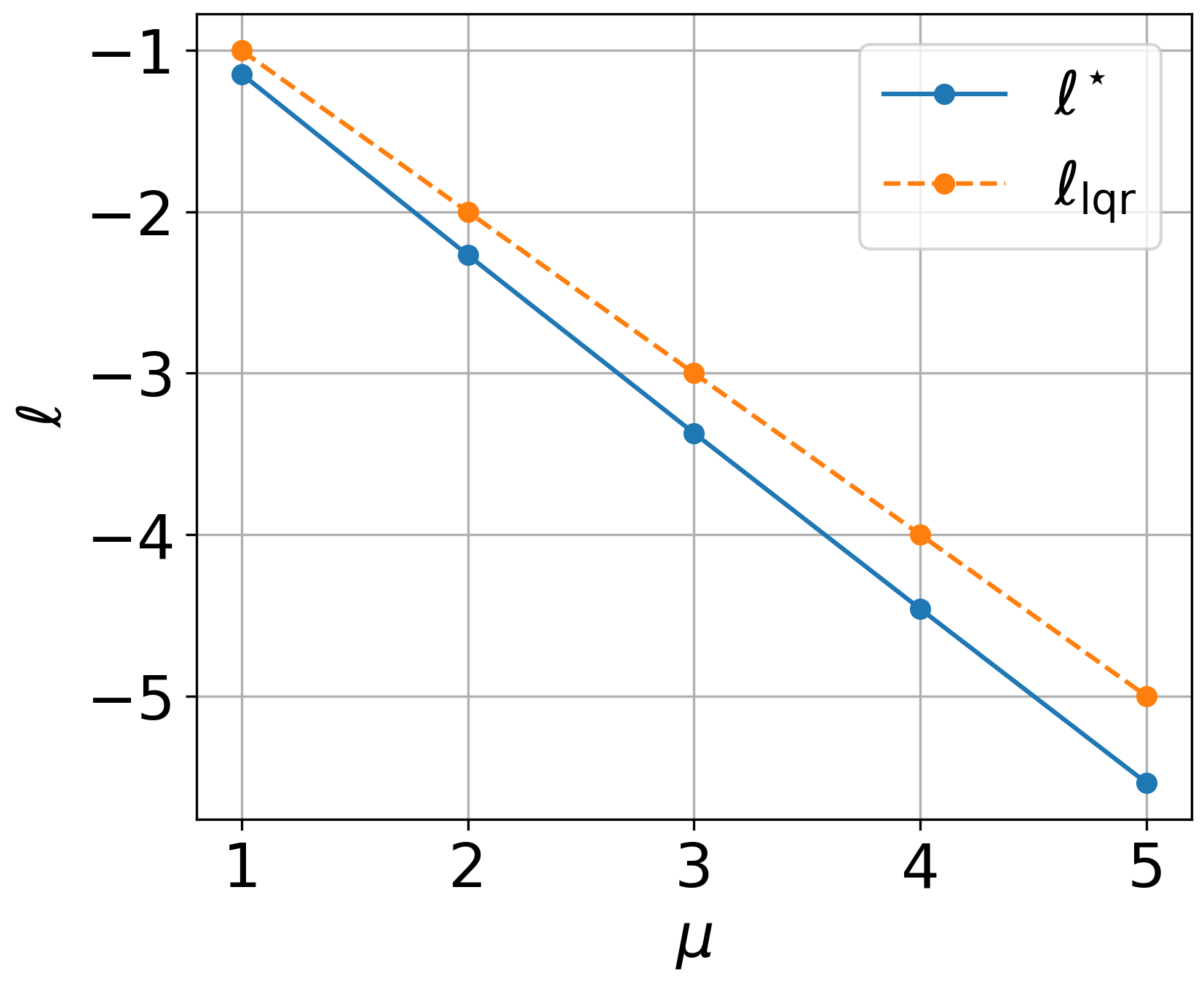}
    \caption{bias $\ell$} \label{fig:example_1d_controller:ell}
  \end{subfigure}

  \begin{subfigure}[b]{0.49\linewidth}
    \centering
    \includegraphics[width=0.99\linewidth]{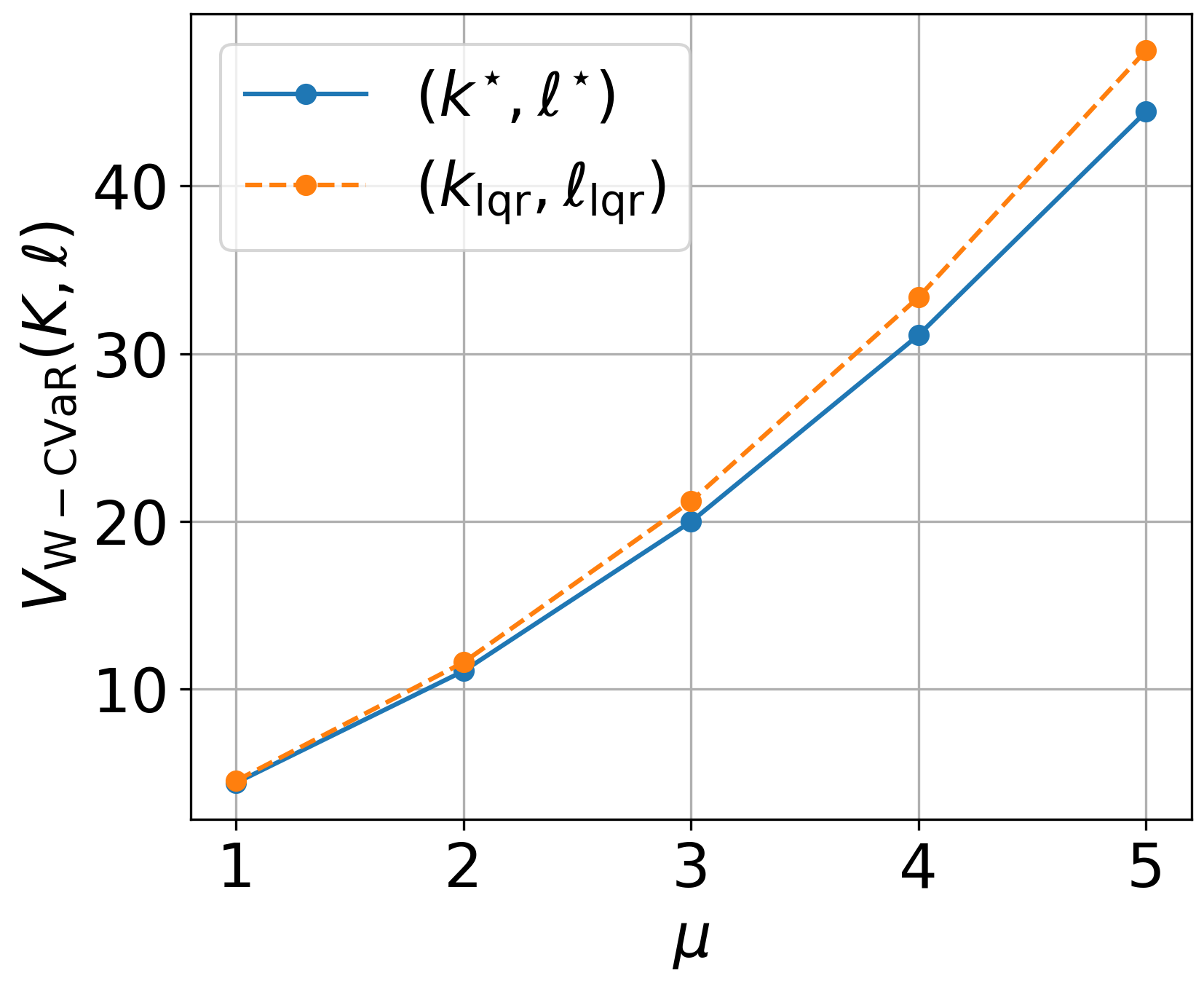}
    \caption{W-CVaR values} \label{fig:example_1d_controller:value}
  \end{subfigure}
  \begin{subfigure}[b]{0.49\linewidth}
    \centering
    \includegraphics[width=0.99\linewidth]{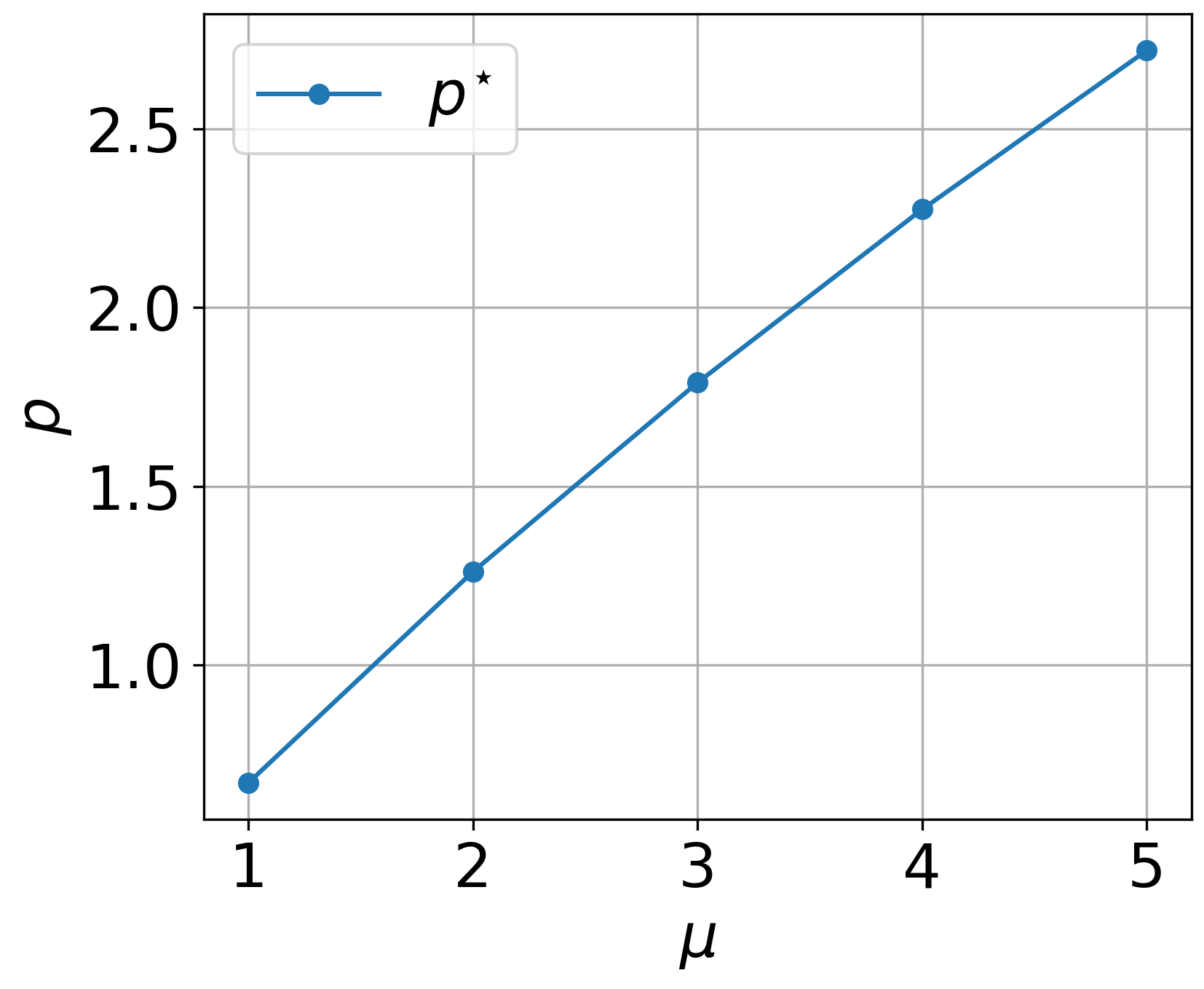}
    \caption{coupling term $p^*$} \label{fig:example_1d_controller:p}
  \end{subfigure}
  \caption{Comparison of $(K^{\star}, \ell^{\star})$ and $(K_{\lqr}, \ell_{\lqr})$.} \label{fig:example_1d_controller}
\end{figure}

It can be observed that the W-CVaR-optimal affine controller becomes increasingly more distinct from the naive affine LQR controller as the mean and covariance of the noise increase (see \Cref{fig:example_1d_controller:k,fig:example_1d_controller:ell}). This justifies that, unlike the LQR control synthesis, the nonzero-mean setting cannot be reduced to the zero-mean setting for risk-sensitive control synthesis. The performance gap in terms of W-CVaR values also increases with $\mu$, indicating a growing advantage of the W-CVaR-optimal affine controller over the naive affine LQR controller as shown in \Cref{fig:example_1d_controller:value}. \Cref{fig:example_1d_controller:p} suggests a potential explanation of these phenomena that echoes the analysis above---as the off-diagonal entry $p$ (in matrix $M$) grows, the coupling between decision variables is strengthened, resulting in a growing discrepancy from the zero-mean setting. 

The above example clarifies that the setting with nonzero-mean noise is fundamentally distinct from the zero-mean case considered in \cite{kishida2023risk}, and also reflects the challenges inherent to the risk-sensitive problem \eqref{eq:problem-main} to be addressed.
  \section{Risk-sensitive Affine Control Synthesis}\label{sec:3-algorithm}
Herein, we first elaborate on the intrinsic challenges in the risk-sensitive problem formulated in (\ref{eq:problem-main}), and then propose our reformulation as a BMI problem. Through this reformulation, we eventually provide an algorithm that solves the problem by alternating between two convex sub-problems.

\subsection{Intrinsic Challenges of the nonzero-mean Problem}\label{sec:3-1-general_mean_challenge}
To compare against the zero-mean case, we write $M$ in block form compatible with $\bOmega$ as
\begin{equation}\label{eq:block_form_M}
  M = \begin{bmatrix}
    M_1 & p \\
    * & m_0
  \end{bmatrix},
\end{equation}
where $M_1 \in \mathcal{S}^{n}_{\succeq 0}$, $p \in \R^{n}$ and $m_0 \in \R$. Consider the special case where $\wmu = \matO$ and $\ell = \matO$ (hence $\bmu = \matO$) studied in \cite{kishida2023risk}, so that the second-order moment matrix $\bOmega$ is decoupled to contain only the diagonal blocks $\bSigma$ and $1$, namely
\begin{equation*}
  \bOmega = \begin{bmatrix}
    \bSigma & \matO \\
    * & 1
  \end{bmatrix}. 
\end{equation*}
In this way, the objective in \eqref{eq:problem-main} also decouples, in the sense that $M_1$ and $m_0$ are completely independent in both \eqref{eq:problem-main:1} and \eqref{eq:problem-main:5}, and $p$ also disappears after simplification. Leveraging this decoupling between $M_1$ and $m_0$, it can then be shown that the optimal solution coincides with the LQR controller  \cite{kishida2023risk}.

We emphasize, however, that the decoupling is only feasible when $\ell = 0$ and $\bmu = \matO$. More concretely, after partitioning $M$ into block form \eqref{eq:block_form_M} in the nonzero-mean setting, we will find that the objective of \eqref{eq:problem-main} becomes
\begin{equation*}
  \alpha + \tfrac{1}{1-\beta}\prn[\big]{\angl{M_1, \bSigma + \bmu \bmu^{\top}} + 2p^{\top} \bmu + m_0},
\end{equation*}
where the coupling between decision variables is significantly stronger since the right-hand side of \eqref{eq:problem-main:5} is no longer block diagonal, and that the objective contains an additional term $p^{\top} \bmu$ that does not admit an immediate lower bound. This strong coupling makes it impossible to optimize for $K$ and $\ell$ separately, justifying the necessity of solving \eqref{eq:problem-main}.

\subsection{Relaxation to a Bilinear Matrix Inequality (BMI) Problem}

As discussed above, a general-purpose nonlinear optimization solver cannot solve \eqref{eq:problem-main} easily due to the following challenges: i) presence of high-order equality constraints \eqref{eq:problem-main:3} and \eqref{eq:problem-main:4}, ii) a mixture of matrix equality and inequality constraints, iii) the substantial coupling between the decision variables in the objective and the constraints. These challenges will be settled via LMI relaxation, resulting in an alternating optimization approach which will be further elaborated.
First, we transform \eqref{eq:problem-main} into a BMI form summarized below.

\begin{theorem}\label{thm:relax_LMI}
  Problem \eqref{eq:problem-main} can be relaxed to the following BMI:
  \begin{subequations}\label{eq:problem-relaxed}
  \begin{align}
    \min_{\alpha, \tM, \tOmega, \tmu, \bSigma, Y, \ell, H}\quad& \alpha + \tfrac{1}{1-\beta} \angl{\tOmega, \tM}, \label{eq:problem-relaxed:1}\\
    \mathrm{s.t.}\quad
      & \tM \in \mathcal{S}^{n+1}_{\succeq 0},~ \bSigma \in \mathcal{S}^{n}_{\succeq 0}, \label{eq:problem-relaxed:2}\\
      & H = A\bSigma+BY \label{eq:problem-relaxed:3}\\
      & \begin{bmatrix}
          \bSigma - \wSigma & H \\
          * & \bSigma  
        \end{bmatrix} \succeq 0, \label{eq:problem-relaxed:4}\\
      & \bSigma\tmu = H  \tmu + B\ell + \wmu, \label{eq:problem-relaxed:5}\\
      & \begin{bmatrix}
          \tM_1 & \tp & Y^{\top} & \bSigma \\
          * & \tm_0+\alpha & \ell^{\top} & \matO \\
          * & * & R^{-1} & \matO \\
          * & * & * & Q^{-1}
        \end{bmatrix} \succeq 0, \label{eq:problem-relaxed:6}\\
      & \begin{bmatrix}
          \tOmega_1  & \tq & \tmu & I\\
          * & \tomega_0 & 1 & 0 \\
          * & * & 1 & 0\\
          * & * & * & \bSigma
        \end{bmatrix} \succeq 0. \label{eq:problem-relaxed:7}
  \end{align}
  \end{subequations}
  Here we write $\tM$ and $\tOmega$ in the following block forms:
  \begin{equation}
    \tM =: \begin{bmatrix}
      \tM_1 & \tp \\
      * & \tm_0
    \end{bmatrix},\quad
    \tOmega =: \begin{bmatrix}
      \tOmega_1 & \tq \\
      * & \tomega_0,
    \end{bmatrix}
  \end{equation}
  where $\tM_1, \tOmega_1 \in \S^n$, $\tp, \tq \in \R^n$, $\tm_0, \tomega_0 \in \R$.
\end{theorem}

\begin{proof}
We start by noticing that we don't need an explicit constraint $K \in \K$, since it is already implicitly enforced in \eqref{eq:problem-main:3}, as the discrete-time Lyapunov equation only has solutions when the closed-loop system is stable (see \Cref{thm:lyapunov_equation}).

To convert the constraints into LMIs, we first need to ``normalize'' the problem via change-of-variables. For this purpose, define the following auxiliary variables:
\begin{gather*}
  \tM := \begin{bmatrix}
  \bSigma \\
  & 1
\end{bmatrix} M \begin{bmatrix}
  \bSigma \\
  & 1
\end{bmatrix},~
\tOmega := \begin{bmatrix}
  \bSigma^{-1} \\
  & 1
\end{bmatrix} \bOmega \begin{bmatrix}
  \bSigma^{-1} \\
  & 1
\end{bmatrix},\\
\tmu := \bSigma^{-1} \bmu,\quad
Y := K \bSigma,\quad
H := A \bSigma + B Y,
\end{gather*}
so that \eqref{eq:problem-main} can be equivalently written as follows:
\begin{subequations}\label{eq:problem-normalized}
\begin{align}
  \min_{\alpha, \tM, \tOmega, \tmu, \bSigma, Y, \ell, H}\quad& \alpha + \tfrac{1}{1-\beta} \angl{\tOmega, \tM} \label{eq:problem-normalized:1}\\
  \mathrm{s.t.}\quad
    & \tM \in \mathcal{S}^{n+1}_{\succeq 0},~ \bSigma \in \mathcal{S}^{n}_{\succeq 0},  \label{eq:problem-normalized:2}\\
    & H = A \bSigma + B Y, \label{eq:problem-normalized:3}\\
    & H \bSigma^{-1} H^{\top} + \wSigma = \bSigma, \label{eq:problem-normalized:4}\\
    & \bSigma\tmu = H \tmu + B\ell + \wmu, \label{eq:problem-normalized:5}\\
    & \tM \succeq \begin{bmatrix}
      \bSigma Q \bSigma + Y^{\top} R Y & Y^{\top} R\ell \\
      * & \ell^{\top} R \ell -\alpha
    \end{bmatrix}, \label{eq:problem-normalized:6}\\
    & \tOmega = \begin{bmatrix}
      \bSigma^{-1} + \tmu \tmu^{\top} & \tmu \\
      * & 1
    \end{bmatrix}, \label{eq:problem-normalized:7}
\end{align}
\end{subequations}

To convert \eqref{eq:problem-normalized:4} into an LMI, we first relax the equality constraint into the following inequality (the direction is selected to avoid unbounded objective)
\begin{equation*}
  H \bSigma^{-1} H^{\top} + \wSigma \preceq \bSigma,
\end{equation*}
which can then be equivalently rewritten as an LMI
\begin{equation}\label{eq:relax_LMI_d}
  \begin{bmatrix}
    \bSigma - \wSigma & H \\
    * & \bSigma  
  \end{bmatrix} \succeq 0
\end{equation}
using Schur complement. Similarly, for \eqref{eq:problem-normalized:6}, we have
\begin{align}
  &\tM \succeq \begin{bmatrix}
      \bSigma Q \bSigma + Y^{\top} R Y & Y^{\top} R\ell \\
      * & \ell^{\top} R \ell -\alpha
    \end{bmatrix} \nonumber\\
  \iff\;& \begin{bmatrix}
    \tM_1 - \bSigma Q \bSigma & \tp \\
    * & \tm_0 + \alpha
  \end{bmatrix} - \begin{bmatrix}
    Y^{\top} \\ \ell^{\top}
  \end{bmatrix} R \begin{bmatrix}
    Y & \ell
  \end{bmatrix} \succeq 0 \nonumber\\
  \iff\;&  \begin{bmatrix}
    \tM_1 - \bSigma Q \bSigma & \tp & Y^{\top} \\
    * & \tm_0 + \alpha & \ell^{\top} \\
    * & * & R^{-1}
  \end{bmatrix} \succeq 0 \nonumber\\
  \iff\;& \begin{bmatrix}
      \tM_1 & \tp & Y^{\top} & \bSigma \\
      * & \tm_0+\alpha & \ell^{\top} & \matO \\
      * & * & R^{-1} & \matO \\
      * & * & * & Q^{-1}
    \end{bmatrix} \succeq 0; \label{eq:relax_LMI_f}
\end{align}
for \eqref{eq:problem-normalized:7}, we first relax it into an inequality (again the direction is selected to avoid unbounded objective), and then we have
\begin{align}
  &\tOmega \succeq \begin{bmatrix}
      \bSigma^{-1} + \tmu \tmu^{\top} & \tmu \\
      * & 1
    \end{bmatrix} \nonumber\\
  \iff\;& \begin{bmatrix}
    \tOmega_1 - \bSigma^{-1} & \tq \\
    * & \tomega_0
  \end{bmatrix} - \begin{bmatrix}
    \tmu \\ 1
  \end{bmatrix} \begin{bmatrix}
    \tmu^{\top} & 1
  \end{bmatrix} \succeq 0 \nonumber\\
  \iff\;& \begin{bmatrix}
      \tOmega_1 - \bSigma^{-1} & \tq & \tmu \\
      * & \tomega_0 & 1 \\
      * & * & 1 \\
    \end{bmatrix} \succeq 0 \nonumber\\
  \iff\;& \begin{bmatrix}
      \tOmega_1  & \tq & \tmu & I\\
      * & \tomega_0 & 1 & 0 \\
      * & * & 1 & 0\\
      * & * & * & \bSigma
    \end{bmatrix} \succeq 0. \label{eq:relax_LMI_g}
\end{align}
Plugging \eqref{eq:relax_LMI_d}, \eqref{eq:relax_LMI_f} and \eqref{eq:relax_LMI_g} into \eqref{eq:problem-normalized} completes the proof.
\end{proof}

\begin{remark}
  Even though all the constraints are relaxed to LMIs, the program itself is still a BMI due to the coupling in the objective function. Furthermore, all the relaxations introduced are actually exact, except for the one in \eqref{eq:problem-relaxed:4} that accounts for the suboptimality of our proposed solution.
\end{remark}

\subsection{An Alternating Optimization Algorithm}\label{sec:algorithm}

So far, we have shown how to relax \eqref{eq:problem-main} to \eqref{eq:problem-relaxed}, where the latter has the favorable property that the LMI constraints are linear in the decision variables. However, the relaxed problem \eqref{eq:problem-relaxed} is still hard to solve since it is still heavily coupled in the following senses: i) the objective \eqref{eq:problem-relaxed:1} involves the product of two decision variables $\tOmega$ and $\tM$, and ii) \eqref{eq:problem-relaxed:5} is coupled and cannot be further relaxed without compromising the invariance property of the stationary mean value. As a result, standard BMI optimization solvers still cannot solve program \eqref{eq:problem-relaxed} directly.

To compute \eqref{eq:problem-relaxed}, we design the following optimization algorithm that alternates between two subsets of variables, $\bxi := (\tOmega, \tmu, \ell)$, and $\btheta := (\alpha, M, \bSigma, Y, H)$, and repetitively executes the following two steps:
\begin{itemize}
    \item \textbf{Step I:}
 we first optimize over $\bxi = (\tOmega, \tmu, \ell)$ using $\btheta^{\star} = (\alpha^{\star}, M^{\star}, \bSigma^{\star}, Y^{\star}, H^{\star})$ obtained in the last iteration:
\begin{subequations}\label{eq:problem-relaxed-stage_1}
\begin{align}
  \min_{\tOmega, \tmu, \ell}\quad& \alpha^{\star} + \tfrac{1}{1-\beta} \angl{\tOmega, \tM^{\star}}, \label{eq:problem-relaxed-stage_1:1}\\
  \mathrm{s.t.}\quad
    & \tOmega \in \mathcal{S}^{n+1}_{\succeq 0}, \label{eq:problem-relaxed-stage_1:2}\\
    & \bSigma^{\star}\tmu = H^{\star}  \tmu + B^{\star}\ell + \wmu, \label{eq:problem-relaxed-stage_1:3}\\
    & \begin{bmatrix}
      \tM^{\star}_1 & \tp^{\star} & (Y^{\star})^{\top} & \bSigma^{\star} \\
      * & \tm^{\star}_0+\alpha^{\star} & \ell^{\top} & \matO \\
      * & * & R^{-1} & \matO \\
      * & * & * & Q^{-1}
    \end{bmatrix} \succeq 0, \label{eq:problem-relaxed-stage_1:4}\\
    & \begin{bmatrix}
      \tOmega_1  & \tq & \tmu & I\\
      * & \tomega_0 & 1 & 0 \\
      * & * & 1 & 0\\
      * & * & * & \bSigma^{\star}
    \end{bmatrix} \succeq 0.\label{eq:problem-relaxed-stage_1:5}
\end{align}
\end{subequations}

\item \textbf{Step II:} we then use $\bxi_i^{\star} = (\tOmega^{\star}, \tmu^{\star}, \ell^{\star})$ obtained in \eqref{eq:problem-relaxed-stage_1} to optimize over $\btheta = (\alpha, M, \bSigma, Y, H)$ as follows:
\begin{subequations}\label{eq:problem-relaxed-stage_2}
\begin{align}
  \min_{\alpha, \tM, \bSigma, Y, H}\quad& \alpha + \tfrac{1}{1-\beta} \angl{\tOmega^{\star}, \tM}, \label{eq:problem-relaxed-stage_2:1}\\
  \mathrm{s.t.}\quad
    & \tM \in \mathcal{S}^{n+1}_{\succeq 0},~ \bSigma \in \mathcal{S}^{n}_{\succeq 0}, \label{eq:problem-relaxed-stage_2:2}\\
    & H = A\bSigma+BY, \label{eq:problem-relaxed-stage_2:3}\\
    & \begin{bmatrix}
        \bSigma - \wSigma & H \\
        * & \bSigma  
      \end{bmatrix} \succeq 0, \label{eq:problem-relaxed-stage_2:4}\\
    & \bSigma\tmu^{\star} = H  \tmu^{\star} + B\ell^{\star} + \wmu, \label{eq:problem-relaxed-stage_2:5}\\
    & \begin{bmatrix}
      \tM_1 & \tp & Y^{\top} & \bSigma \\
      * & \tm_0+\alpha & (\ell^{\star})^{\top} & \matO \\
      * & * & R^{-1} & \matO \\
      * & * & * & Q^{-1}
    \end{bmatrix} \succeq 0, \label{eq:problem-relaxed-stage_2:6}\\
    & \begin{bmatrix}
      \tOmega^{\star}_1  & \tq^{\star} & \tmu^{\star} & I\\
      * & \tomega^{\star}_0 & 1 & 0 \\
      * & * & 1 & 0\\
      * & * & * & \bSigma
    \end{bmatrix} \succeq 0. \label{eq:problem-relaxed-stage_2:7}
\end{align}
\end{subequations}
\end{itemize}

We point out that both \eqref{eq:problem-relaxed-stage_1} and \eqref{eq:problem-relaxed-stage_2} are now convex in the decision variables, and are thus solvable using generic convex optimization solvers (e.g., CVX \cite{cvx}).

\begin{remark}[Numerical considerations]
  For numerical stability, practical implementation may also include additional bounds to avoid extra large search spaces (e.g., $\tr(\tM) \leq B_{M}$ and $\tr(\tOmega) \leq B_{\varOmega}$ for sufficiently large $B_{M}$ and $B_{\varOmega}$), as well as margins to avoid numerical issues (e.g., replacing $0$ with $\varepsilon I$ in LMI constraints with appropriately small $\varepsilon$). In the first iteration, $\btheta^{\star}_0 = (\alpha_0^{\star}, M_0^{\star}, \bSigma_0^{\star}, Y_0^{\star}, H_0^{\star})$ needs to be properly initialized to ensure feasibility of the program.
\end{remark}

\subsection{Convergence Guarantee}\label{sec:convergence}

We conclude this section by proving a convergence guarantee for the proposed alternating optimization algorithm, which in turn is based on the boundedness of the objective. For this purpose, we first prove the following lemma.

\begin{lemma}\label{thm:bounded_objective}
  For any matrices $M, \bOmega \in \S^{n+1}_{\succeq 0}$ satisfying
  \begin{equation*}
    M \succeq \begin{bmatrix}
      \bSigma + \bmu \bmu^{\top} & \bmu \\
      * & 1
    \end{bmatrix},~
    \bOmega \succeq \begin{bmatrix}
      Q + K^{\top} R K & K^{\top} R \ell \\
      * & \ell^{\top} R \ell - \alpha
    \end{bmatrix}
  \end{equation*}
  for some $Q \in \S^n_{\succ 0}$, $R \in \S^n_{\succ 0}$, $\bmu \in \R^n$, $\bSigma \in \S^n_{\succeq 0}$, $K \in \R^{m \times n}$, $\ell \in \R^m$, we have
  \begin{equation*}
    \alpha + \tfrac{1}{1-\beta} \angl{\bOmega, M} \geq 0.
  \end{equation*}
\end{lemma}

\begin{proof}
  When $\alpha \geq 0$, it is obvious that
  \begin{equation*}
    \alpha + \tfrac{1}{1-\beta} \angl{\bOmega, M}
    \geq \alpha
    \geq 0,
  \end{equation*}
  where we use the fact that $\bOmega, M \succeq 0$ implies $\angl{\bOmega, M} \geq 0$.
  
  Otherwise, when $\alpha < 0$, we have
  \begin{subequations}\label{eq:bounded_objective:e2}
  \begin{align}
    &\alpha + \tfrac{1}{1-\beta} \angl{\bOmega, M} \label{eq:bounded_objective:e2-1}\\
    \geq{}& \alpha + \tfrac{1}{1-\beta} \angl*{\begin{bsmallmatrix}
      \bSigma + \bmu \bmu^{\top} & \bmu \\
      * & 1
    \end{bsmallmatrix}, \begin{bsmallmatrix}
      Q + K^{\top} R K & K^{\top} R \ell \\
      * & \ell^{\top} R \ell - \alpha
    \end{bsmallmatrix}} \label{eq:bounded_objective:e2-2}\\
    ={}& \alpha + \tfrac{1}{1-\beta} \prn*{ \angl{\bSigma + \bmu\bmu^{\top}, Q} + \angl*{\bOmega, \begin{bsmallmatrix}
      K^{\top} \\ \ell^{\top}
    \end{bsmallmatrix} R \begin{bsmallmatrix}
      K & \ell
    \end{bsmallmatrix}} - \alpha } \label{eq:bounded_objective:e2-3}\\
    \geq{}& - \tfrac{\beta}{1-\beta} \alpha
    \geq 0. \label{eq:bounded_objective:e2-4}
  \end{align}
  \end{subequations}
  Here in \eqref{eq:bounded_objective:e2-2} we use the fact that $\angl{X_1, P} \geq \angl{X_2, P}$ if $P \succeq 0$ and $X_1 \succeq X_2$; in \eqref{eq:bounded_objective:e2-3} we apply the linearity of trace; in \eqref{eq:bounded_objective:e2-4} we use the positive-semidefiniteness of all the matrices involved and $\beta \in (0,1)$. This completes the proof.
\end{proof}

\begin{theorem}\label{thm:convergence}
  The value sequence produced by the alternating optimization algorithm converges.
\end{theorem}

\begin{proof}
  Denote by $J(\bxi, \btheta) := \alpha + \tfrac{1}{1-\beta} \angl{\tOmega, \tM}$ the value of \eqref{eq:problem-relaxed}, and by $\varXi_i$, $\varTheta_i$ the constraint set in \eqref{eq:problem-relaxed-stage_1}, \eqref{eq:problem-relaxed-stage_2}, respectively. Let $\JI_i$, $\JII_i$ denote the optimal values obtained in Step I and II of the $i$\tsup{th} iteration, respecitvely. Then the algorithm ensures that
  \begin{align*}
    \JI_i &= \min_{\bxi_i \in \varXi_i} J(\bxi_i, \btheta^{\star}_{i-1})
    \leq J(\bxi^{\star}_{i-1}, \btheta^{\star}_{i-1})
    = \JII_{i-1},\\
    \JII_i &= \min_{\btheta_i \in \varTheta_i} J(\bxi^{\star}_i, \btheta_i)
    \leq J(\bxi^{\star}_i, \btheta^{\star}_{i-1})
    = \JI_i.
  \end{align*}
  Therefore, the value sequence $\brac{\JI_1, \JII_1, \JI_2, \JII_2, \ldots}$ is monotone decreasing. Meanwhile, by \Cref{thm:bounded_objective} we also have
  \begin{equation*}
    \JII_i, \JI_{i+1} \geq 0,\quad \forall i \in \N.
  \end{equation*}
  Finally, we can conclude the convergence of the value sequence by the Monotone Convergence Theorem.
\end{proof}
  \section{Simulations}\label{sec:4-simulations}

In this section, we evaluate the numerical performance of the proposed alternating optimization algorithm in two representative settings, Inverted Pendulum \cite{hespanha2018linear} and Multi-zone Heating, Ventilation, and Air Conditioning (HVAC) Thermal Control \cite{li2021distributed}. We illustrate the convergence behavior of our algorithm, as well as the suboptimality gaps arising from the relaxation and superority over the naive affine LQR controller. We use both the \textit{W-CVaR value} $V_{\wcvar}(K,\ell) := \cP_{\bmu, \bSigma}\bcvar[c_{K,\ell}(\bx)]$ and the \textit{LQR value} $V_{\lqr}(K,\ell)$ (see \Cref{sec:appdx-LQR_value} for details) as metrics.

The alternating optimization algorithm is implemented with CVX using MOSEK optimizer, iteratievely solving \eqref{eq:problem-relaxed-stage_1} and \eqref{eq:problem-relaxed-stage_2} until convergence to obtain the solution controller $u = K_{\sol} x + \ell_{\sol}$. For low-dimensional systems, the W-CVaR-optimal affine controller $u = K^{\star} x + \ell^{\star}$ is approximated via manually-calibrated grid search (to the accuracy of $10^{-3}$). The initialization of $\btheta^{\star}_0$ is obtained by solving \eqref{eq:WCVaR_equivalence} with $(K,\ell)$ assigned as the naive affine LQR controller $(K_{\lqr}, \ell_{\lqr})$. 

Details of the system parameters are deferred to \Cref{sec:apdx-simulation_setting}.

\subsection{Inverted Pendulum}\label{sec:simulation-Pendulum}

We first consider the Inverted Pendulum setting. A typical convergence behavior of our algorithm is shown in \Cref{fig:simulation_pendulum-convergence}. \Cref{fig:simulation_pendulum-convergence:K,fig:simulation_pendulum-convergence:ell} show the magnitude of step-wise changes $\Delta K$ and $\Delta \ell$ over time, where both gradually converge to 0. The improvement in terms of W-CVaR value in \Cref{fig:simulation_pendulum-convergence:value} is made legible by a constant shift and displayed on a log scale. All figures illustrate a pattern where the algorithm converges within 500 iterations, agreeing with the theoretical guarantee.

\begin{figure}[htb]
  \centering
  \begin{subfigure}[b]{0.49\linewidth}
    \centering
    \includegraphics[width=0.99\linewidth]{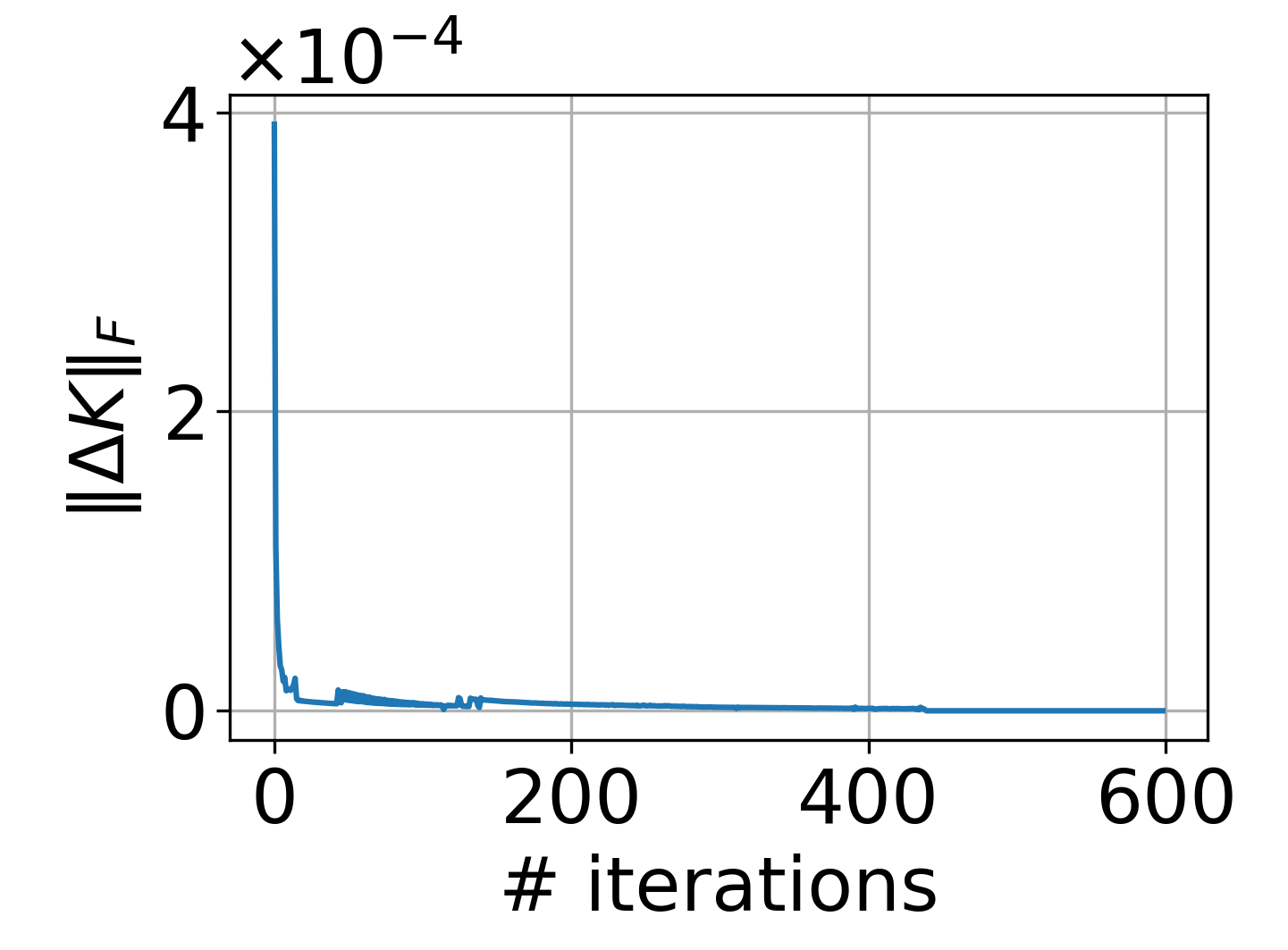}
    \caption{$\norm{\Delta K}_{\mathrm{F}}$}
    \label{fig:simulation_pendulum-convergence:K}
  \end{subfigure}
  \begin{subfigure}[b]{0.49\linewidth}
    \centering
    \includegraphics[width=0.99\linewidth]{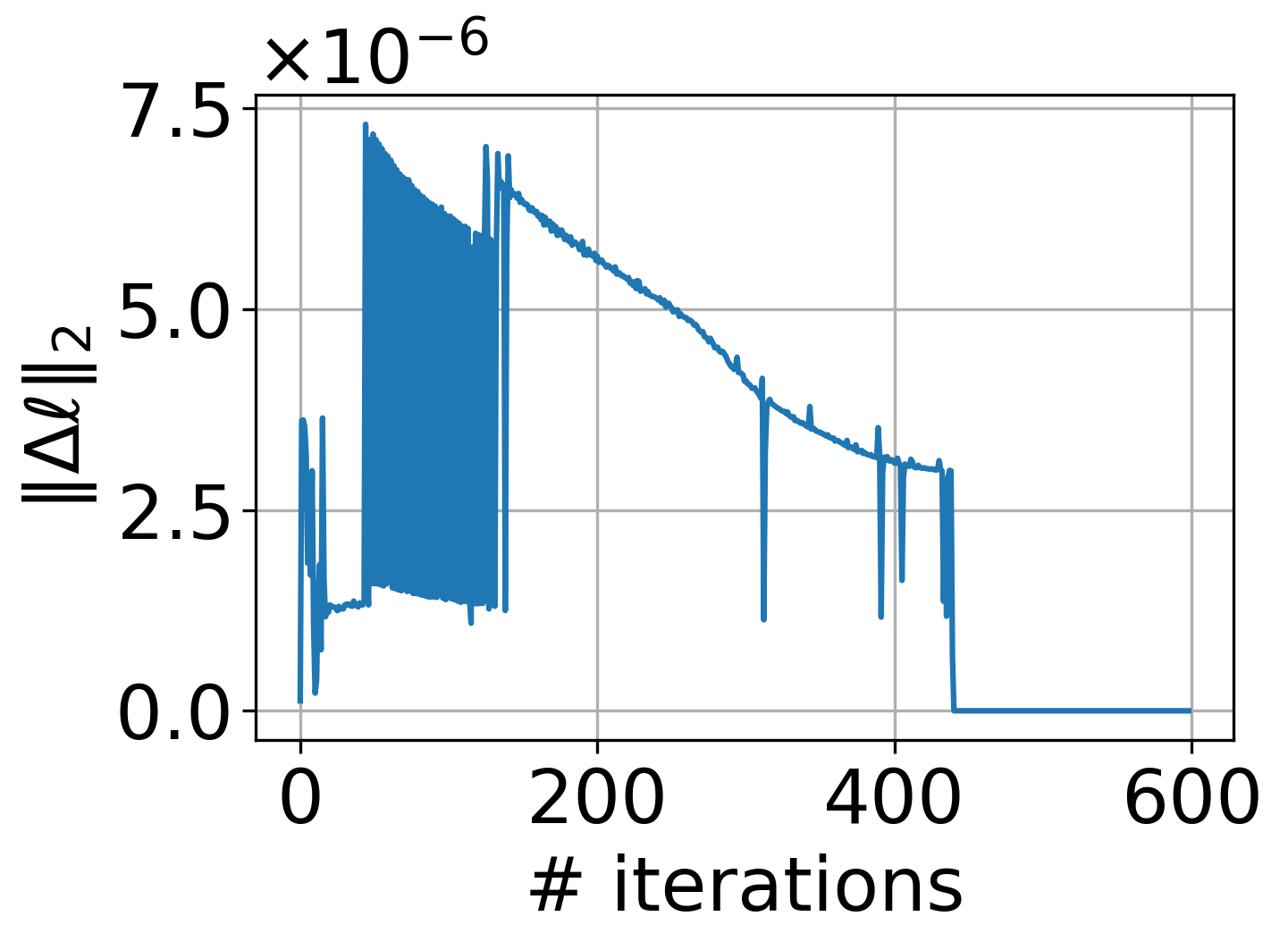}
    \caption{$\norm{\Delta \ell}_{2}$}
    \label{fig:simulation_pendulum-convergence:ell}
  \end{subfigure}

  \begin{subfigure}[b]{0.84\linewidth}
    \centering
    \includegraphics[width=0.99\linewidth]{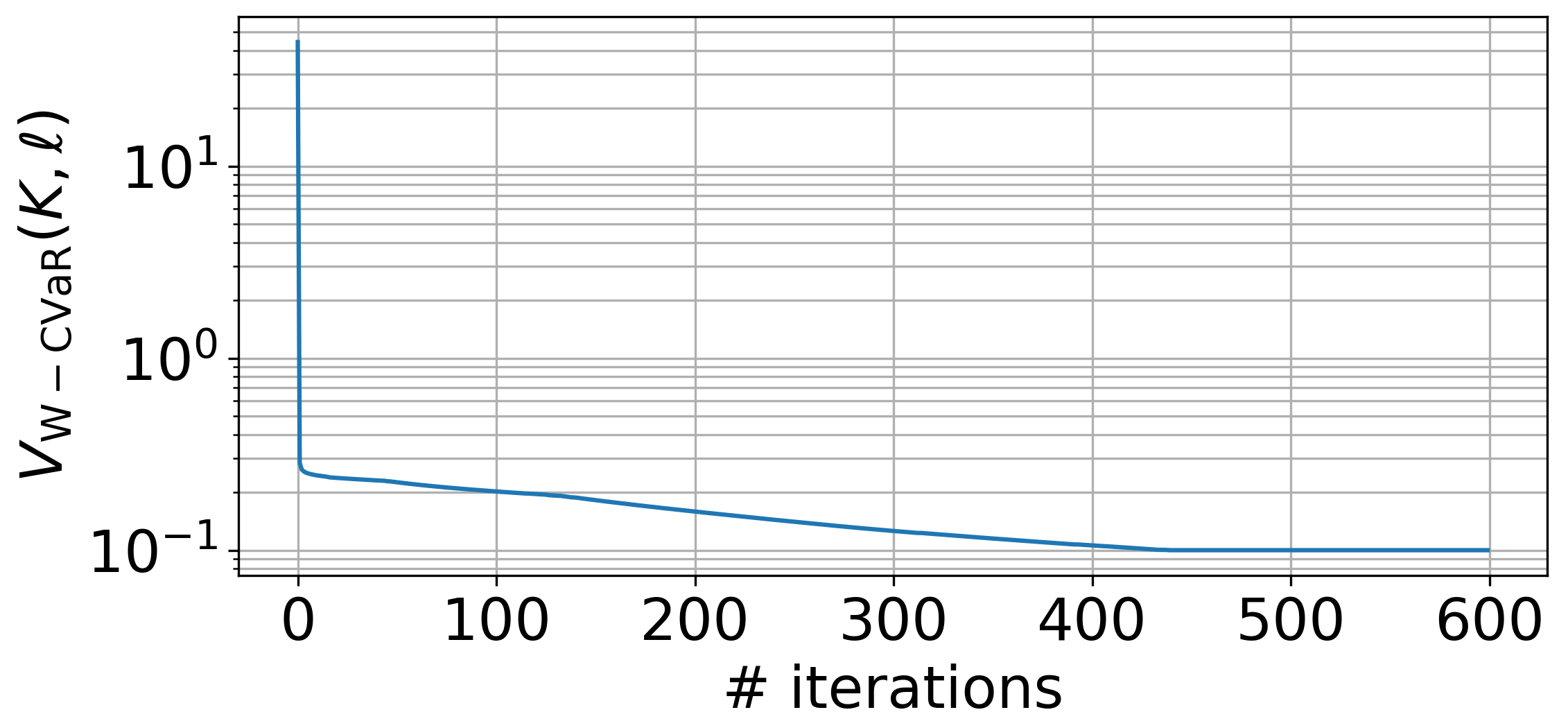}
    \caption{W-CVaR value (log-scale, constant-shifted)}
    \label{fig:simulation_pendulum-convergence:value}
  \end{subfigure}
  \caption{Convergence of the algorithm ($L = 1$)}
  \label{fig:simulation_pendulum-convergence}
\end{figure}

The solution $(K_{\sol}, \ell_{\sol})$ is then compared against two baselines in terms of W-CVaR and LQR values: the LQR controller $u = K_{\lqr} x$ (which turns out to perform better than the naive affine LQR controller), and the W-CVaR-optimal controller $(K^{\star}, \ell^{\star})$. \Cref{fig:simulation_pendulum-value:WCVaR} shows the W-CVaR values at different system parameters $L \in \set{1,2,3,4,5}$. It can be observed that the solution always outperforms the LQR baseline, and the performance gap increases as the system becomes less controllable (i.e., as $L$ grows). Although the solution is suboptimal as compared to the W-CVaR-optimal controller, the suboptimality gap decreases as $L$ grows, indicating a higher benefit of the proposed algorithm in uncertain environments. Finally, \Cref{fig:simulation_pendulum-value:WCVaR} showcases that risk-sensitiveness is obtained at a price, as the solution induces a (moderately) higher LQR value as compared to the LQR baseline.

\begin{figure}[htb]
  \centering
  \begin{subfigure}[b]{0.49\linewidth}
    \centering
    \includegraphics[width=0.99\linewidth]{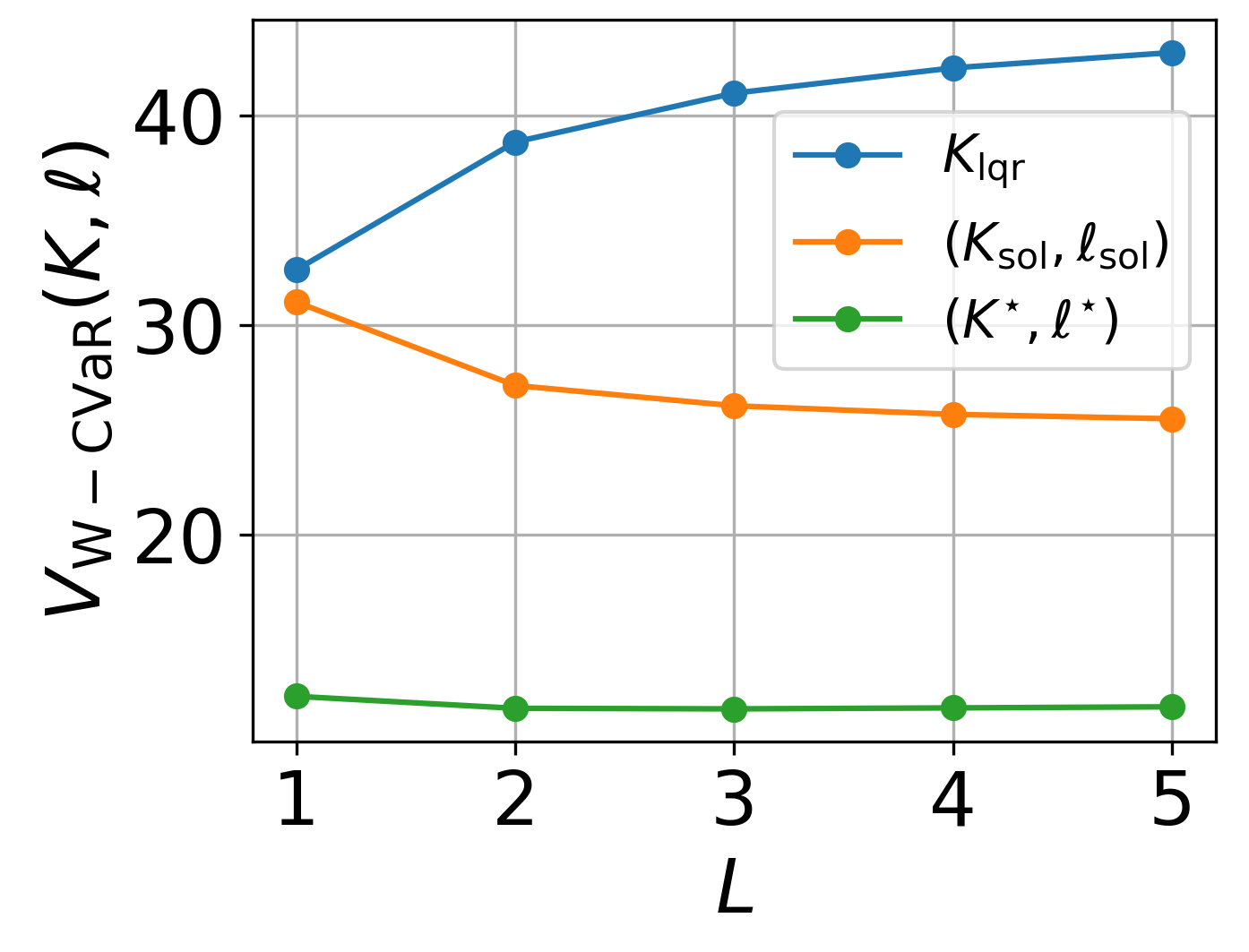}
    \caption{W-CVaR values}
    \label{fig:simulation_pendulum-value:WCVaR}
  \end{subfigure}
  \begin{subfigure}[b]{0.49\linewidth}
    \centering
    \includegraphics[width=0.99\linewidth]{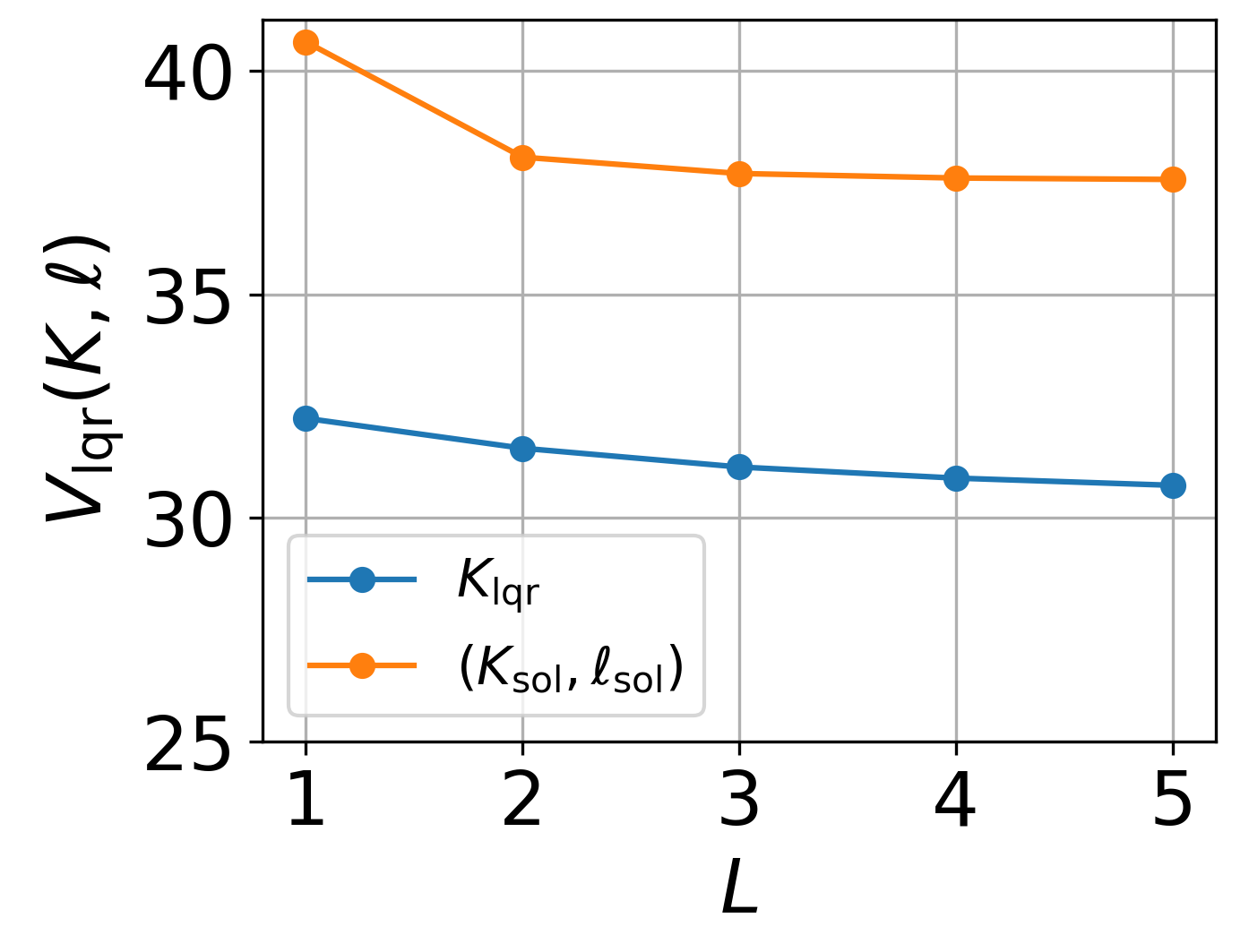}
    \caption{LQR values}
    \label{fig:simulation_pendulum-value:LQR}
  \end{subfigure}
  \caption{Performance of the solution controller against baselines for Inverted Pendulum with different pole lengths $L$.}
  \label{fig:simulation_pendulum-value}
\end{figure}

\subsection{The Multi-zone HVAC Thermal Control}\label{sec:simulation-HVAC}

\begin{wrapfigure}{r}{1.6cm}
    \centering
    \vspace{-25pt}
    \begin{tikzpicture}
      \foreach \x in {0pt,16pt,32pt,48pt}{
        \draw[black, line width=0.5pt] (0pt, \x) -- (48pt, \x);
        \draw[black, line width=0.5pt] (\x,0pt) -- (\x,48pt);
      }
      \node at (8pt,40pt) {\small 1};
      \node at (24pt,40pt) {\small 2};
      \node at (40pt,40pt) {\small 3};
      \node at (8pt,24pt) {\small 4};
      \node at (24pt,24pt) {\small 5};
      \node at (40pt,24pt) {\small 6};
      \node at (8pt,8pt) {\small 7};
      \node at (24pt,8pt) {\small 8};
      \node at (40pt,8pt) {\small 9};
    \end{tikzpicture}
    \caption*{HVAC grid ($d = 3$).}
    \vspace{-10pt}
\end{wrapfigure}
We also evaluate our method in a setting that regulates thermal control in multi-zone HVAC systems modelled as an LTI system \cite{li2021distributed}. We consider an HVAC system with $d^2$ zones arranged in a $d$-by-$d$ grid.

The solution $(K_{\sol}, \ell_{\sol})$ is again compared against the naive affine LQR controller baseline. \Cref{fig:simulation_HVAC:W-CVaR} shows that the proposed algorithm produces a risk-sensitive controller that significantly outperforms the baseline naive LQR controller in terms of W-CVaR values, and that the edge keeps growing as the number of zones increases. Furthermore, \Cref{fig:simulation_HVAC:LQR} illustrates that such benefit comes at a negligible loss of LQR values, which can be attributed to the invertibility of the input matrix $B$ (see \Cref{sec:apdx-simulation_setting}).

\begin{figure}[htb]
  \centering
  \begin{subfigure}[b]{0.49\linewidth}
    \centering
    \includegraphics[width=0.99\linewidth]{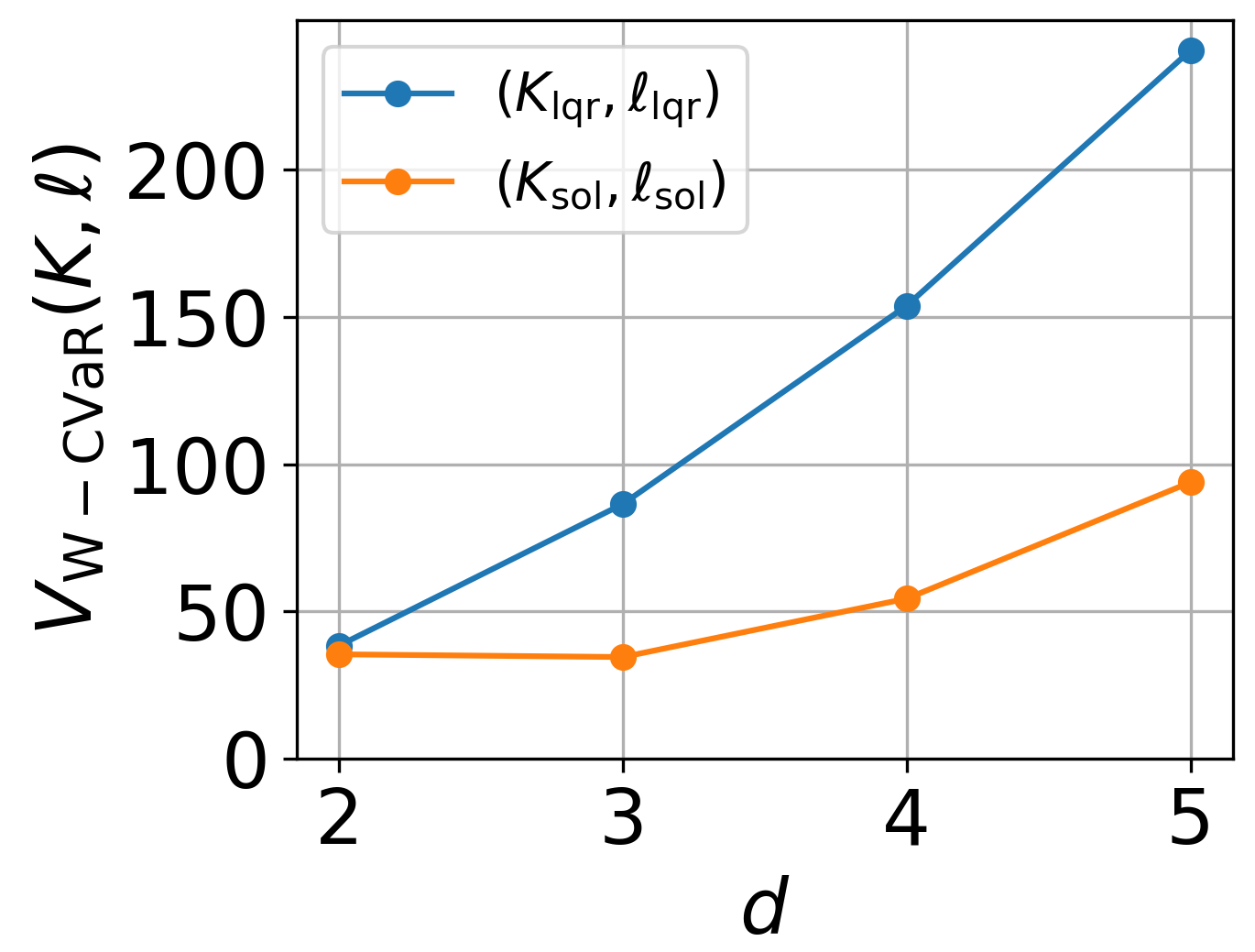}
    \caption{W-CVaR values}
    \label{fig:simulation_HVAC:W-CVaR}
  \end{subfigure}
  \begin{subfigure}[b]{0.49\linewidth}
    \centering
    \includegraphics[width=0.99\linewidth]{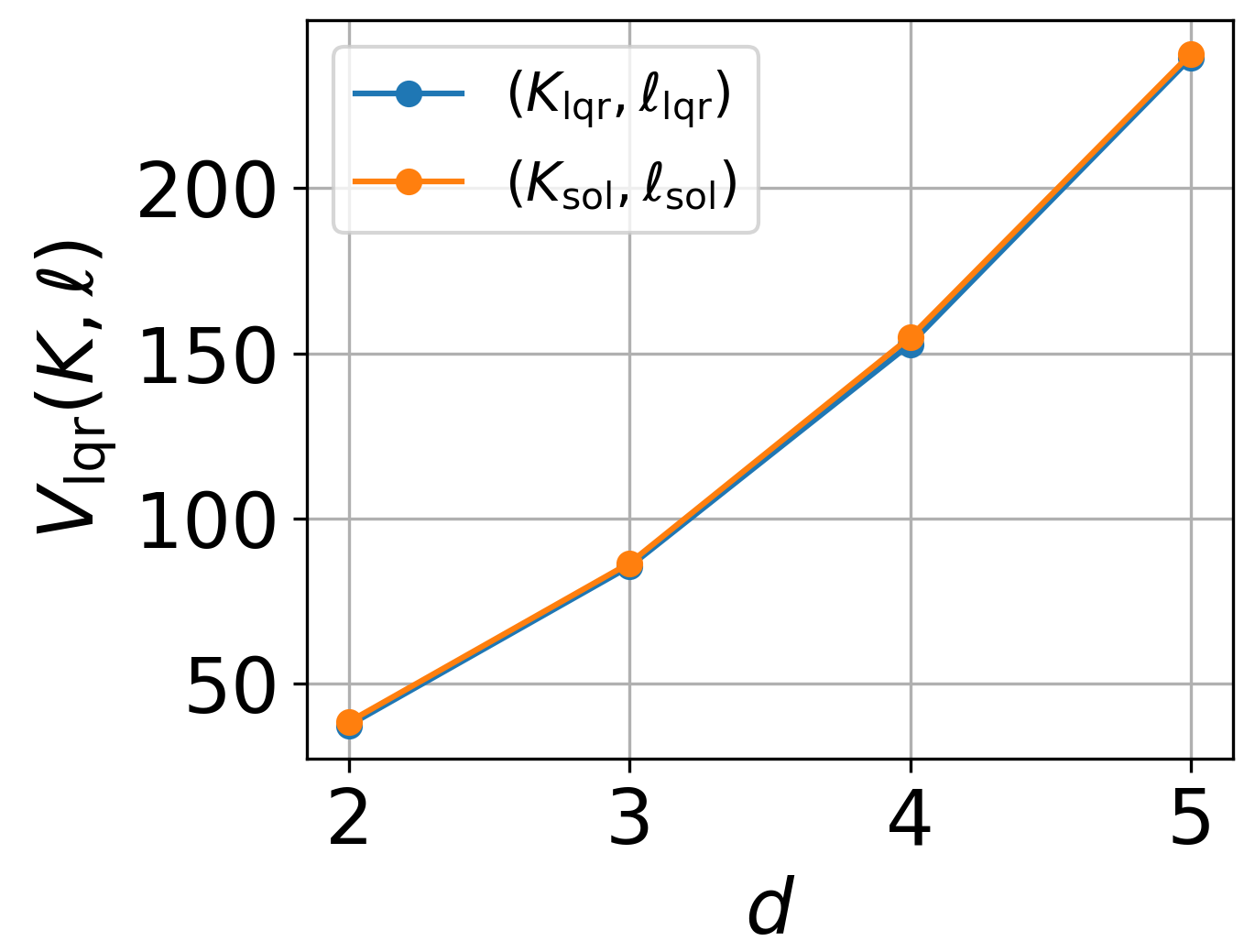}
    \caption{LQR values}
    \label{fig:simulation_HVAC:LQR}
  \end{subfigure}
  \caption{Performance of the solution controller against baselines for an HVAC grid system of size $d \times d$.}
  \label{fig:simulation_HVAC}
\end{figure}
  \section{Conclusions}

This paper extended risk-sensitive control synthesis in stationary LTI systems to affine controllers, accounting for nonzero-mean noise under the Worst-case CVaR (W-CVaR) risk measure. Unlike the zero-mean case where the optimal controller coincides with the LQR optimal one, we discovered that the linear and bias terms of the affine controller are strongly coupled, preventing decoupling of the optimization. By reformulating the problem as a Bilinear Matrix Inequality (BMI), we developed an alternating optimization algorithm that converges to a suboptimal solution. Simulations on Inverted Pendulum and HVAC systems demonstrated that our approach effectively produces risk-sensitive (suboptimal) controllers for systems under uncertainty, with increasing advantage over the naive affine LQR controller in terms of the W-CVaR measure as the dimension and complexity of the problem grow.

Although our algorithm relies solely on solving convex programs, they involve instances of semi-definite programming (SDPs), and thus may not scale efficiently to very large dimensions. Additionally, while our proposed solution is close to the W-CVaR-optimal one in certain cases, the suboptimality gap has not yet been analytically characterized. It is left as future work to address these limitations.

  \bibliographystyle{IEEEtran}
  \bibliography{ref}


  \vspace{2em}
  \appendix
  \section{Proofs}\label{sec:A1-proofs}

\subsection{Characterizing the Limiting Moments}\label{sec:apdx-limiting_distribution_proof}

For the sake of notations, write $\mu_t \in \R^n$ and $\varSigma_t \in \S^n_{\succeq 0}$ for the mean and covariance of $x_t$, where $t \in \N$. We will use the shorthand notation $A_K := A + BK$ throughout the appendix. Note that $\rho(A_K) < 1$ since we restrict $K \in \K$.

\begin{proof}[Proof of \Cref{thm:limiting_distribution}]
Note that the closed-loop dynamics is
\begin{equation*}
  x_{t+1} = A_K x_t + B \ell + w_t,~ \forall t \in \N.
\end{equation*}
Take the telescoping sum, and we have
\begin{equation*}
  x_t = A_K^t x_0 + \sum_{\tau=0}^{t-1} A_K^{t-\tau-1} (B\ell + w_{\tau}),
\end{equation*}
which implies that $\mu_t$ and $\varSigma_t$ should be
\begin{align*}
  \mu_t &= A_K^t \mu_0 + \sum_{\tau=0}^{t-1} A_K^{t-\tau-1} (B\ell + \wmu) \nonumber\\
    &= A_K^t \mu_0 + \sum_{\tau=0}^{t-1} A_K^{\tau} (B\ell + \wmu),\\
  \varSigma_t &= A_K^t \varSigma_0 (A_K^t)^{\top} + \sum_{\tau=0}^{t-1} A_K^{t-\tau-1} \wSigma (A_K^{t-\tau-1})^{\top}\nonumber\\
    &= A_K^t \varSigma_0 (A_K^t)^{\top} + \sum_{\tau=0}^{t-1} A_K^{\tau} \wSigma (A_K^{\tau})^{\top}.
\end{align*}
The above equations are well-defined since $\rho(A_K) < 1$. Further, since $\norm{A_K^t \mu_0} \leq \norm{A_K}^t \norm{\mu_0} \to 0$, the sequence of means eventually converge to
\begin{equation}\label{eq:limiting_mean}
  \bmu := \lim_{t \to \infty} \mu_t = \sum_{\tau = 0}^{\infty} A_K^{\tau} (B\ell + \wmu).
\end{equation}
Similarly, since $\norm*{A_K^t \varSigma_0 (A_K^t)^{\top}} \leq \norm{\varSigma_0} \norm{A_K^t}^2 \to 0$, the sequence of covariances eventually converge to
\begin{equation}\label{eq:limiting_covariance}
  \bSigma := \lim_{t \to \infty} \varSigma_t = \sum_{\tau=0}^{\infty} A_K^{\tau} \wSigma (A_K^{\tau})^{\top}.
\end{equation}
Therefore, the first and second moments in the limit are uniquely determined by $(\bmu, \bSigma)$. It can then be directly verified that \eqref{eq:limiting_mean} and \eqref{eq:limiting_covariance} are exactly the solutions to \eqref{eq:stationary_distribution}.
\end{proof}

\begin{remark}
  Note that \eqref{eq:stationary_distribution} describes exactly the \textit{invariant moments} under the closed-loop dynamics, which coincides with the (unique) \textit{limiting moments} regardless of $(\mu_0, \varSigma_0)$.
\end{remark}

We end this section by taking a closer look at \eqref{eq:stationary_distribution:Sigma}, which is a \textit{discrete-time Lyapunov equation} that is closely related to the stability of the system through the following lemma.

\begin{lemma}[\cite{hespanha2018linear}]\label{thm:lyapunov_equation}
  A matrix $A$ is Schur stable (i.e., $\rho(A) < 1$) if and only if for any $Q \succ 0$, there exists a unique solution $P$ to the discrete-time Lyapunov equation
  \begin{equation*}
    A^{\top} P A - P = -Q.
  \end{equation*}
  Further, the solution is positive-definite, and can be explicitly expressed in series form as
  \(
    P = \sum_{\tau=0}^{\infty} (A^{\tau})^{\top} Q A^{\tau}
  \).
\end{lemma}

\subsection{Technical Lemmas Regarding W-CVaR Measure}\label{sec:apdx-WCVaR_lemma}

For the paper to be self-contained, we include the technical lemmas we need regarding the W-CVaR measure. The following lemma provides an equivalent characterization of W-CVaR.

\begin{lemma}[\cite{zymler2013distributionally}]
  For any random variable $L: \varXi \to \R$ defined on sample space $\varXi$, and any family of distributions $\cP \subseteq \Delta(\varXi)$, the W-CVaR with respect to $\cP$ is given by
  \begin{align}\label{eq:WCVaR_characterize}
    &\cP\bcvar_{\beta}[L(\xi)] ={} \nonumber\\
    &\quad \min_{\alpha \in \R} \brac[\Big]{\alpha + \tfrac{1}{1-\beta} \sup_{\P \in \cP} \E[\xi \sim \P]{(L(\xi)-\alpha)^+}}.
  \end{align}
\end{lemma}

In observation of the specific form of \eqref{eq:WCVaR_characterize}, the following lemma comes in handy when we want to eliminate the expectation from the optimization formulation of W-CVaR for the specific family $\cP_{\mu,\varSigma}$ (as defined in \Cref{sec:settings-WCVaR})

\begin{lemma}[\cite{zymler2013distributionally}]
  For any random variable $L: \varXi \to \R$ defined on sample space $\varXi$, with respect to the family $\cP_{\mu,\varSigma}$ of distributions with prescribed mean $\mu$ and covariance $\varSigma$,
  \begin{align}\label{eq:worst_case_moment_dual}
    &\sup_{\P \in \cP_{\mu,\varSigma}} \E[\P \sim \cP_{\mu,\varSigma}]{(L(\xi))^+} ={}\\
    &\quad \min_{M \in \mathcal{S}^{n+1}_{\succeq 0}} \brac*{\angl{\varOmega, M} \;\middle|\; \begin{bsmallmatrix}
      \xi \\ 1
    \end{bsmallmatrix}^{\top} M \begin{bsmallmatrix}
      \xi \\ 1
    \end{bsmallmatrix} \geq L(\xi),~ \forall \xi \in \R^n}, \nonumber
  \end{align}
  where
  \(
    \varOmega := \begin{bsmallmatrix}
      \varSigma + \mu \mu^{\top} & \mu \\
      \mu^{\top} & 1
    \end{bsmallmatrix}
  \) denotes the second-order moment matrix of any distribution in $\cP_{\mu,\varSigma}$.
\end{lemma}

In the special case where $\varXi = \R^n$ and $L(\xi) = \xi^{\top} P \xi + q^{\top} \xi + r$ is quadratic in $\xi$, \eqref{eq:worst_case_moment_dual} can be substituted into \eqref{eq:WCVaR_characterize} to obtain the following characterization of W-CVaR with a quadratic objective, which is the main technical tool we use.

\begin{lemma}[characterization of worst-case CVaR with a quadratic objective] \label{thm:quadratic_worst_CVaR_characterization}
  For $L(\xi) = \xi^{\top} P \xi + q^{\top} \xi + r$, we have
  \begin{align*}
    &\cP_{\mu,\varSigma}\bcvar_{\beta}[L(\xi)] ={}\\
    &\quad \min_{\begin{subarray}{c} \alpha \in \R, M \in \mathcal{S}^{n+1}_{\succeq 0}  \end{subarray}} \brac*{\alpha + \tfrac{1}{1-\beta} \angl{\varOmega, M} \;\middle|\; M \succeq \begin{bsmallmatrix}
      P & \frac{1}{2} q \\
      \frac{1}{2} q^{\top} & r - \alpha
    \end{bsmallmatrix}}. \nonumber
  \end{align*}
\end{lemma}

\subsection{Detailed Settings of the Simulation}\label{sec:apdx-simulation_setting}

\begin{wrapfigure}{r}{1.5cm}
    \centering
    \vspace{-20pt}
    \begin{tikzpicture}
      \filldraw[black] (0pt, 0pt) circle (1pt);
      \draw[black, line width=0.5pt] (-20:3pt) arc (-20:-200:3pt) -- ([shift=(160:3pt)]70:48pt) -- ([shift=(-20:3pt)]70:48pt) -- (-20:3pt);
      \draw[black, fill=white, line width=1.0pt] (70:56pt) circle (8pt);
      \draw[black, line width=0.5pt, dashed] (0pt,-10pt) -- (0pt, 50pt);
      \node at (70:56pt) {\small $r$};
      \node[left=6pt] at (70:56pt) {\small $G$};
      \node at ([shift=(-20:8pt)]70:24pt) {\small $L$};

      \draw[black, line width=0.5pt, ->, >={Stealth}] (120:10pt) arc (120:390:10pt);
      \node at (12pt, -12pt) {\small $J$};
    \end{tikzpicture}
    \vspace{-20pt}
\end{wrapfigure}
\textbf{Inverted Pendulum.} We linearize the well-known Inverted Pendulum dynamics around the upright position \cite{hespanha2018linear}, i.e.,
\begin{equation*}
  x_{t+1} = \underbrace{\begin{bmatrix}
    0 & 1 \\
    -\frac{GL}{J} & -\frac{r}{J}
  \end{bmatrix}}_{A} x_t + \underbrace{\begin{bmatrix}
    0 \\ -\frac{L}{J}
  \end{bmatrix}}_{B} u_t + w_t.
\end{equation*}
Here $G$ is the weight of the pendulum, $L$ is the length of the pole, $r$ is the radius of the pendulum, and $J$ is its moment of inertia. For simplicity, we fix $G = J = r = 1$, and leave $L$ as a variable parameter. Finally, set $Q = \diag(2,1)$ and $R = \begin{bmatrix} 2 \end{bmatrix}$.

\boldtitle{HVAC system.} The system dynamics can be described as

\begin{small}
\begin{align*}
  x_{t+1} = \underbrace{ \begin{bsmallmatrix}
    1 - \frac{n\delta}{\upsilon \zeta_{11}} & \frac{\delta}{\upsilon \zeta_{12}} & \cdots & \frac{\delta}{\upsilon \zeta_{1n}} \\
    \frac{\delta}{\upsilon \zeta_{21}} & 1 - \frac{n\delta}{\upsilon \zeta_{22}} & \sddots & \svdots \\
    \svdots & \sddots & \sddots & \frac{\delta}{\upsilon \zeta_{n-1,n}} \\
    \frac{\delta}{\upsilon \zeta_{n1}} & \cdots & \frac{\delta}{\upsilon \zeta_{n,n-1}} & 1 - \frac{n\delta}{\upsilon \zeta_{nn}} \\
  \end{bsmallmatrix} }_{A} x_t + \underbrace{ \tfrac{\delta}{\upsilon} I }_{B} u_t + w_t,
\end{align*}
\end{small}

\noindent where $w_t := \frac{\sqrt{\delta}}{\upsilon} d_t + \frac{\delta}{\upsilon} \pi + \frac{\delta \theta^{\circ}}{\upsilon \zeta}$.  Detailed explanations of the dynamics can be found in Section 6.1 of \cite{li2021distributed}. In our implementation, the parameters are selected similar to those in \cite{li2021distributed}: $\delta = \qty{60}{\second}$, $\theta^{\circ} = \qty{30}{\celsius}$, $\pi = \qty{1}{\kilo\watt}$, $\upsilon = \qty{400}{\kilo\joule\per\celsius}$, $\zeta = \qty{1}{\celsius\per\kilo\watt}$, $\alpha = 0.01$. Suppose the disturbance $d_t$ admits mean $0$ and covariance $I$, so that the overall process noise $w_t$ has nonzero mean $\wmu = \frac{\delta}{\upsilon} \pi_i + \frac{\delta \theta^{\circ}}{\upsilon \zeta}$ and covariance $\wSigma = \frac{\delta}{\upsilon^2} I$.

\subsection{Characterizing the LQR Value}\label{sec:appdx-LQR_value}

The following lemma characterizes the expected average LQR value $V_{\lqr}(K, \ell)$ of an affine controller $u = Kx + \ell$.

\begin{lemma}
  Given an affine controller $u = Kx + \ell$, we have
  \begin{alignfit}[0.95\linewidth]
    V_{\lqr}(K,\ell)
    :={}& \mathop{\lim\sup}_{T \to \infty} \frac{1}{T} \sum_{t=1}^{T} \E{x_t^{\top} Q x_t + u^{\top} R u_t} \\
    ={}& \tr\prn[\big]{ P_K (\wSigma + \nu\nu^{\top} ) } + 2 \ell^{\top} (RK + B^{\top} P_K A_K) \bmu \\
       &\quad + 2\wmu^{\top} P_K A_K \bmu + \ell^{\top} R \ell.
  \end{alignfit}
  Here $P_K$ is the solution to the discrete-time Lyapunov equation $A_K^{\top} P_K A_K + Q_K = P_K$, where we define shorthand notations $Q_K := Q + K^{\top} R K$ and $\nu := B\ell + \wmu$.
\end{lemma}

\begin{proof}
  Denote by $\varPhi_t := \Es{x_t x_t^{\top}}$ the second \textit{raw} moment of $x_t$. Let $\varGamma_T := \sum_{t=0}^{T-1} \varPhi_t$, $S_T := \sum_{t=0}^{T-1} \mu_t$. Note that
  \begin{alignfit}
    \varPhi_t &= \E{(A_K x_{t-1} + B\ell + w_{t-1}) (A_K x_{t-1} + B\ell + w_{t-1})^{\top}} \\
    &= A_K \varPhi_{t-1} A_K^{\top} + (\nu \nu^{\top} + \wSigma) + \nu \mu_{t-1}^{\top} A_K^{\top} + A_K \mu_{t-1} \nu^{\top}.
  \end{alignfit}
  Take the telescoping sum over $t$, from $0$ to $T$, we have
  \begin{equationfit}[0.95\linewidth]
    \varGamma_{T+1} - \varPhi_0 = A_K \varGamma_T A_K^{\top} + T (\nu \nu^{\top} + \wSigma) + \nu S_T^{\top} A_K^{\top} + A_K S_T \nu^{\top}.
  \end{equationfit}
  Divide both sides by $T$ and let $T \to \infty$, we conclude that
  \begin{equationfit}
    \bGamma_{\infty} = A_K \bGamma_{\infty} A_K^{\top} + (\nu \nu^{\top} + \wSigma) + \nu (A_K \bmu)^{\top} + (A_K \bmu) \nu^{\top},
  \end{equationfit}
  where we define $\bGamma_{\infty} := \lim\limits_{T \to \infty} \frac{1}{T} \varGamma_T$ and $\bS_{\infty} := \lim\limits_{T \to \infty} \frac{1}{T} S_T$, and use $\bS_{\infty} = \bmu$. Meanwhile, the LQR value is given by
  \begin{alignfit}[0.95\linewidth]
    V_{\lqr}(K,\ell)
    &= \mathop{\lim\sup}_{T \to \infty} \frac{1}{T} \E{\sum_{t=1}^{T} \prn*{ x_t^{\top} Q_K x_t + 2 \ell^{\top} RK x_t + \ell^{\top} R \ell} } \\
    &= \mathop{\lim\sup}_{T \to \infty} \frac{1}{T} \prn*{ \tr(Q_K \varGamma_T) + 2 \ell^{\top} RK S_T} + \ell^{\top} R \ell \\
    &= \tr(Q_K \bGamma_{\infty}) + 2 \ell^{\top} R K\bmu + \ell^{\top} R \ell.
  \end{alignfit}
  Finally, by linearity and cyclic property of trace, we have
  \begin{alignfit}[0.9\linewidth]
    \tr(Q_K \bGamma_{\infty})
    &= \tr\prn[\big]{(P_K - A_K^{\top} P_K A_K) \bGamma_{\infty})} \\
    &= \tr(P_K \bGamma_{\infty}) - \tr(A_K^{\top} P_K A_K \bGamma_{\infty}) \\
    &= \tr(P_K \bGamma_{\infty}) - \tr(P_K A_K \bGamma_{\infty} A_K^{\top}) \\
    &= \tr\prn[\big]{P_K (\bGamma_{\infty} - A_K \bGamma_{\infty} A_K^{\top})} \\
    &= \tr\prn[\big]{ P_K ( \nu \nu^{\top} + \wSigma + \nu (A_K \bmu)^{\top} + (A_K \bmu) \nu^{\top} ) } \\
    &= \tr\prn[\big]{ P_K (\wSigma + \nu \nu^{\top}) } + 2 \nu^{\top} P_K A_K \bmu.
  \end{alignfit}
  Finally, we plug the above equation along with $\nu = B\ell + \wmu$ back into the expression of $V_{\lqr}$. This completes the proof.
\end{proof}
  \vfill

\end{document}